\documentclass[11pt]{article} 
\usepackage[sort&compress,numbers]{natbib}
\usepackage{mathtools}

\usepackage{amsmath,amssymb,latexsym,amsopn,amsfonts}
\usepackage[boxed,ruled,vlined,linesnumbered,noresetcount]{algorithm2e}
\usepackage{url,fullpage}
\usepackage{framed,microtype}
\usepackage[normalem]{ulem}
\usepackage{array,xspace}
\usepackage{graphicx}
\usepackage{color}
\usepackage[dvipsnames]{xcolor}
\newcommand{\ems}[1]{\textcolor{blue}{#1}} 
\definecolor{oskar_green}{rgb}{0.0, 0.5, 0.0}
\newcommand{\OL}[1]{\textcolor{oskar_green}{[\textbf{OL: #1}]}}
\newcommand{\algSize}{small}

\usepackage[printwatermark]{xwatermark}

\newcommand{\bigO}{\mathcal{O}\xspace}
\newcommand{\remove}[1]{}

\newcommand{\txS}{\mathit{txObsS}\xspace}
\newcommand{\txSJ}{\mathit{txObsSJ}\xspace}
\newcommand{\rxS}{\mathit{rxObsS}\xspace}
\newcommand{\rxSJ}{\mathit{rxObsSJ}\xspace}
\newcommand{\maxS}{\mathsf{maxSeq}\xspace}
\newcommand{\obsolete}{\mathsf{obsolete}\xspace}

\newcommand{\minTxObstSeq}{\mathsf{minTxObsS}\xspace}
\newcommand{\msgSet}{\mathit{buffer}\xspace}

\newcommand{\etal}{\emph{et al.}\xspace}
\newcommand{\eg}{\emph{e.g.,}\xspace}
\newcommand{\ie}{\emph{i.e.,}\xspace}

\usepackage{etoolbox}
\patchcmd{\thebibliography}{\chapter*}{\section*}{}{}

\newtheorem{remark}{Remark}[section]
\newtheorem{theorem}{Theorem}[section]
\newtheorem{lemma}[theorem]{Lemma}
\newtheorem{definition}{Definition}[section]

\newtheorem{claim}[theorem]{Claim}
\newcommand{\true}{\mathsf{True}\xspace}

\newcommand{\false}{\mathsf{False}\xspace}

\newenvironment{claimProof}[1]{\par\noindent\underline{Proof:}\space#1}{\hfill $\blacksquare$}

\newcommand{\sP}{\mathcal{P}\xspace}

\newcommand{\N}{\mathbb{N}\xspace}
\newcommand{\capacity}{\mathsf{channelCapacity}\xspace}
\newcommand{\buffCapacity}{\mathsf{bufferUnitSize}\xspace}

\usepackage{caption}

\newenvironment{proof}{\noindent \textbf{Proof.}}{\hfill$\blacksquare$}
\renewenvironment{claimProof}{\noindent\textbf{Proof of claim.}}{\hfill$\Box$}

\usepackage[utf8]{inputenc}

\begin{document}
\setcounter{footnote}{2}
%
\title{Self-stabilizing Uniform Reliable Broadcast\\~\Large{(preliminary version)}}


\author{Oskar Lundstr\"om~\footnote{Department of Computer Science and Engineering, Chalmers University of Technology, Gothenburg, SE-412 96, Sweden, E-mail: \texttt{osklunds@student.chalmers.se}.} \and Michel Raynal~\footnote{Institut Universitaire de France IRISA, ISTIC Universit\'{e} de Rennes, Rennes Cedex, 35042, France,  and Department of Computing, Polytechnic University, Hong Kong, E-mail: \texttt{michel.raynal@irisa.fr}.} \and Elad M.\ Schiller~\footnote{Department of Computer Science and Engineering, Chalmers University of Technology, Gothenburg, SE-412 96, Sweden, E-mail: \texttt{elad@chalmers.se}.}}


\date{}

\maketitle

\begin{abstract}
	We study a well-known communication abstraction called \emph{Uniform Reliable Broadcast} (URB). URB is central in the design and implementation of fault-tolerant distributed systems, as many non-trivial fault-tolerant distributed applications require communication with provable guarantees on message deliveries. Our study focuses on fault-tolerant implementations for time-free message-passing systems that are prone to node-failures. Moreover, we aim at the design of an even more robust communication abstraction. We do so through the lenses of \emph{self-stabilization}---a very strong notion of fault-tolerance. In addition to node and communication failures, self-stabilizing algorithms can recover after the occurrence of \emph{arbitrary transient faults}; these faults represent any violation of the assumptions according to which the system was designed to operate (as long as the algorithm code stays intact).   
	
	This work proposes the first self-stabilizing URB solution for time-free message-passing systems that are prone to node-failures. The proposed algorithm has an $\bigO(\buffCapacity)$ stabilization time (in terms of asynchronous cycles) from arbitrary transient faults, where $\buffCapacity$ is a predefined constant that can be set according to the available memory. Moreover, the communication costs of our algorithm are similar to the ones of the non-self-stabilizing state-of-the-art. The main differences are that our proposal considers repeated gossiping of $\bigO(1)$ bits messages and deals with bounded space (which is a prerequisite for self-stabilization). Specifically, each node needs to store up to $\buffCapacity \cdot n$ records and each record is of size $\bigO(\nu+n \log n)$ bits, where $n$ is the number of nodes in the system and $\nu$ is the number of bits needed to encode a single URB instance.	 
\end{abstract}

\thispagestyle{empty}

\section{Introduction}
\label{sec:intro}
We propose a self-stabilizing implementation of a communication abstraction called \emph{Uniform Reliable Broadcast} (URB) for time-free message-passing systems whose nodes may fail-stop. 

\smallskip

\noindent \textbf{Context and Motivation.~~} 
Fault-tolerant distributed systems are known to be hard to design and verify. Such complex challenges can be facilitated by high-level communication primitives. These high-level primitives can be based on low-level ones, such as the one that allows nodes to send a message to only one other node at a time. When an algorithm wishes to broadcast message $m$ to all nodes, it can send $m$ individually to every other node. Note that if the sender fails during this broadcast, it can be the case that only some of the nodes have received $m$. Even in the presence of network-level support for broadcasting or multicasting, failures can cause similar inconsistencies. To the end of simplifying the design of fault-tolerant distributed algorithms, such inconsistencies need to be avoided.

The literature has a large number of examples that show how fault-tolerant broadcasts can significantly simplify the development of fault-tolerant distributed systems via State Machine Replication~\cite{DBLP:journals/cn/Lamport78,DBLP:journals/csur/Schneider90}, Atomic Commitment~\cite{DBLP:conf/hase/Raynal97}, Virtual Synchrony~\cite{DBLP:journals/spe/Birman99} and Set-Constrained Delivery Broadcast~\cite{DBLP:conf/icdcn/ImbsMPR18}, to name a few. The weakest variance, named \emph{Reliable Broadcast} (RB), lets all non-failing nodes agree on the set of delivered messages. Stronger RB variants specify additional requirements on the delivery order. Such requirements can simplify the design of fault-tolerant distributed consensus, which allows reaching, despite failures, a common decision based on distributed inputs. Consensus algorithms and RB are closely related problems~\cite{hadzilacos1994modular,DBLP:books/sp/Raynal18}, which have been studied for more than three decades. 


\smallskip

\noindent \textbf{Task description.~~} 
\emph{Uniform Reliable Broadcast} (URB) is a variance of the reliable broadcast problem, which requires that if a node delivers a message, then all non-failing nodes also deliver this message~\cite{hadzilacos1994modular}. The task specifications consider an operation for URB broadcasting of message $m$ and an event of URB delivery of message $m$. The requirements include URB-validity, \ie there is no spontaneous creation or alteration of URB messages, URB-integrity, \ie there is no duplication of URB messages, as well as URB-termination, \ie if the broadcasting node is non-faulty, or if at least one receiver URB-delivers a message, then all non-failing nodes URB-deliver that message. Note that the URB-termination property considers both faulty and non-faulty receivers. This is the reason why this type of reliable broadcast is named \emph{uniform}. 
This work considers a URB implementation that is \emph{quiescent} in the sense that every URB operation incurs a finite number of messages. Moreover, our implementation uses a bounded amount of local memory.

\smallskip

\noindent \textbf{Fault Model.~~} 
We consider a time-free (a.k.a asynchronous) message-passing system that has no guarantees on the communication delay. Moreover, there is no notion of global (or universal) clocks and we do not assume that the algorithm can explicitly access the local clock (or timeout mechanisms). Our fault model includes $(i)$ detectable fail-stop failures of nodes, and $(ii)$ communication failures, such as packet omission, duplication, and reordering. In addition to the failures captured in our model, we also aim to recover from \emph{arbitrary transient faults}, \ie any temporary violation of assumptions according to which the system and network were designed to operate, \eg the corruption of control variables, such as the program counter, packet payload, and operation indices, which are responsible for the correct operation of the studied system, or operational assumptions, such as that the network cannot be partitioned for long periods. Since the occurrence of these failures can be arbitrarily combined, we assume that these transient faults can alter the system state in unpredictable ways. In particular, when modeling the system, we assume that these violations bring the system to an arbitrary state from which a \emph{self-stabilizing algorithm} should recover the system. 

\smallskip

\noindent \textbf{Related Work.~~} 
The studied problem can be traced back to Hadzilacos and Toueg~\cite{hadzilacos1994modular} who consider asynchronous message-passing,  where nodes may fail. They solved several variants to the studied problem with respect to the delivery order, \eg FIFO (first in, first out), CO (causal order), and TO (total order). They also showed that TO-URB and consensus have the same computability power in the context above. Here we focus only on the basic version of URB. To the end of satisfying the quiescent property, we consider a more advanced model, see the remark in~\cite[Section~4.2.1]{DBLP:books/sp/Raynal18}. 
%
%
For a detailed presentation of existing non-self-stabilizing URB solutions and their applications, we refer the reader to~\cite{DBLP:conf/wdag/AguileraCT97,DBLP:books/sp/Raynal18}. (Due to the page limit, Section~\ref{sec:back} of the Appendix brings some of these details.)
%
%
We follow the design criteria of self-stabilization, which was proposed by Dijkstra~\cite{DBLP:journals/cacm/Dijkstra74} and detailed in~\cite{DBLP:books/mit/Dolev2000,DBLP:series/synthesis/2019Altisen}. Dela{\"{e}}t \etal~\cite{DBLP:journals/jpdc/DelaetDNT10} present a self-stabilizing algorithm for the propagation of information with feedback (PIF) that can be the basis for implementing a self-stabilizing URB. However, Dela{\"{e}}t \etal do not consider node failures~\cite[Section 6]{DBLP:journals/jpdc/DelaetDNT10}. To the best of our knowledge, there is no self-stabilizing algorithm that solves the studied problem for the studied fault-model.  

%
%
%

\smallskip

\noindent \textbf{Contributions.~~} We present an important module for dependable distributed systems:  a self-stabilizing algorithm for Uniform Reliable Broadcast (URB) for time-free message-passing systems that are prone to node failures. To the best of our knowledge, we are the first to provide a broad fault model that includes detectable fail-stop failures, communication failures, such as packet omission, duplication, and reordering as well as arbitrary transient faults. The latter can model any violation of the assumptions according to which the system was designed to operate (as long as the algorithm code stays intact). 

The stabilization time of the proposed solution is in $\bigO(\buffCapacity)$ (in terms of asynchronous cycles), where $\buffCapacity$ is a predefined constant that can be set according to the available local memory. Our solution uses only a bounded amount of space, which is a prerequisite for self-stabilization. Specifically, each node needs to store up to $\buffCapacity \cdot n$ records and each record is of size $\bigO(\nu+n \log n)$ bits, where $n$ is the number of nodes in the system and $\nu$ is the number of bits needed to encode a single URB instance. Moreover, the communication costs of our algorithm are similar to the ones of the non-self-stabilizing state-of-the-art. The main difference is that our proposal considers repeated gossiping of $\bigO(1)$ bits messages. 

\smallskip

\noindent \textbf{Organization.} We state our system settings in Section~\ref{sec:sys}. Section~\ref{sec:back} includes a brief overview of some of the earlier ideas that have led to the proposed solution. Our self-stabilizing algorithm is proposed in Section~\ref{sec:basic}; it considers unbounded counters. The correctness proof appears in Section~\ref{sec:qurbCorr}. We explain how to bound the counters of the proposed self-stabilizing algorithm in Section~\ref{sec:bounded}. We conclude in Section~\ref{sec:disc}.

%

\smallskip

\section{System settings}
\label{sec:sys}
We consider a time-free message-passing system that has no guarantees on the communication delay. Moreover, there is no notion of global (or universal) clocks and the algorithm cannot explicitly access the local clock (or timeout mechanisms). The system consists of a set, $\sP$, of $n$ crash-prone nodes (or processors) with unique identifiers. Any pair of nodes $p_i,p_j \in \sP$ have access to a bidirectional communication channel, $channel_{j,i}$, that, at any time, has at most $\capacity \in \N$ packets on transit from $p_j$ to $p_i$ (this assumption is due to a well-known impossibility~\cite[Chapter 3.2]{DBLP:books/mit/Dolev2000}). 

%
%
Our analysis considers the \emph{interleaving model}~\cite{DBLP:books/mit/Dolev2000}, in which the node's program is a sequence of \emph{(atomic) steps}. Each step starts with an internal computation and finishes with a single communication operation, \ie a message $send$ or $receive$. The \emph{state}, $s_i$, of node $p_i \in \sP$ includes all of $p_i$'s variables and $channel_{j,i}$. The term \emph{system state} (or configuration) refers to the tuple $c = (s_1, s_2, \cdots,  s_n)$. We define an \emph{execution (or run)} $R={c[0],a[0],c[1],a[1],\ldots}$ as an alternating sequence of system states $c[x]$ and steps $a[x]$, such that each $c[x+1]$, except for the starting one, $c[0]$, is obtained from $c[x]$ by $a[x]$'s execution. 

\subsection{Task specifications}
\label{sec:spec}
The set of \emph{legal executions} ($LE$) refers to all the executions in which the requirements of the task $T$ hold. In this work, $T_{\text{URB}}$ denotes the task of Uniform Reliable Broadcast (URB) and $LE_{\text{URB}}$ denotes the set of executions in which the system fulfills $T_{\text{URB}}$'s requirements, which Definition~\ref{def:URB} specifies. Definition~\ref{def:URB} considers the operation, $\mathsf{urbBroadcast}(m)$, and the event $\mathsf{urbDeliver}(m)$. When processor $p_i \in \sP$ URB-broadcasts message $m$, it does so by calling $\mathsf{urbBroadcast}(m)$. The specifications assume that every broadcasted message is unique, say, by associating a message identity, \ie the pair $(\mathit{sender}~\mathit{identifier},~ \mathit{sequence}~\mathit{number})$, where the sequence number is an (integer) index that is locally generated by the sender.

\begin{definition}[Uniform Reliable Broadcast~\cite{DBLP:books/sp/Raynal18}]
	\label{def:URB}
	Let $R$ be a system execution. We say that the system demonstrates in $R$ a construction of the URB communication abstraction if the validity, integrity, and termination requirements are satisfied.  
	\begin{itemize}
		\item \textbf{Validity.~~} Suppose that $p_i$ URB-delivers message $m$ in step $a_i \in R$ with $p_j$ as a sender. There is a step $a_j\in R$ that appears in $R$ before $a_i$ in which $p_j$ URB-broadcasts $m$.
		
		\item \textbf{Integrity.~~} $R$ includes at most one step in which processor $p_i$ URB-delivers message $m$.
		
		\item \textbf{Termination.~~} Suppose that a non-faulty $p_i$ takes a step in $R$ that URB-broadcasts or URB-delivers message $m$. Each non-faulty $p_j \in \sP$ URB-delivers $m$ during $R$.
		
	\end{itemize}
\end{definition}

The URB implementation considered in this paper also satisfies the quiescent property (in a self-stabilizing manner). Our implementation uses $\mathsf{MSG}$ and $\mathsf{MSGack}$ messages for conveying information added to the system via $\mathsf{urbBroadcast}$ operations. We say that execution $R$ satisfies the \emph{quiescent} property if every URB-broadcast message that was USB-delivered incurs a finite number of $\mathsf{MSG}$ and $\mathsf{MSGack}$ messages. 
We note that the quiescent property does not consider all the messages that the proposed solution uses. Specifically, we use $\mathsf{GOSSIP}$ messages of constant size that the algorithm sends repeatedly. We note that self-stabilizing systems can never stop sending messages, because if they did, it would not be possible for the system to recover from transient faults~\cite[Chapter 2.3]{DBLP:books/mit/Dolev2000}. 


\subsection{The Fault Model and Self-stabilization}
We model a failure occurrence as a step that the environment takes rather than the algorithm. 


\begin{figure*}[t!]
		\begin{center}
			\begin{\algSize}
		\begin{tabular}{llll}
	\cline{2-3}
	\multicolumn{1}{l|}{}                    & \multicolumn{2}{l|}{~~~~~~~~~~~~~~~~~~~~~~~~~~~~~~~~~~~\textbf{Frequency}}                                                                               \\ \hline
	\multicolumn{1}{|l|}{\textbf{Duration}}  & \multicolumn{1}{l|}{\textit{Rare}}                          & \multicolumn{1}{l|}{\textit{Not rare}}                     \\ \hline
	\multicolumn{1}{|l|}{}                   & \multicolumn{1}{l|}{Any violation of the assumptions according to}         & \multicolumn{1}{l|}{Packet failures: omissions,}        \\
	\multicolumn{1}{|l|}{\textit{Transient}}          & \multicolumn{1}{l|}{which the system operates (but the code stays}           & \multicolumn{1}{l|}{duplications, reordering} \\
	\multicolumn{1}{|l|}{\textit{}}          & \multicolumn{1}{l|}{intact). This can result in any state corruption.}    & \multicolumn{1}{l|}{(assuming fair communications).}      \\
	\hline
	\multicolumn{1}{|l|}{\textit{Permanent}} & \multicolumn{2}{l|}{~~~~~~~~~~~~~~~~~~~~~~~~~~Detectable fail-stop  failures.}                                                                               \\ \hline
	\vspace*{0.25em}
\end{tabular}
			\end{\algSize}

		\hspace*{-0.5em}\includegraphics[clip=true,scale=0.75]{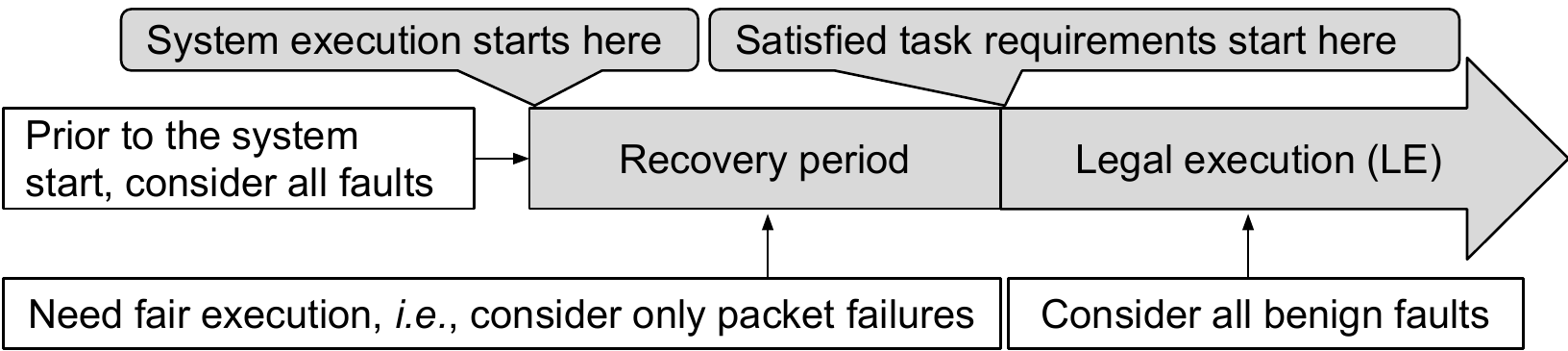}\\
		\end{center}
		\caption{\label{fig:self-stab-SDN}{The table above details our fault model and the chart illustrates when each fault set is relevant. The chart's gray shapes represent the system execution, and the white boxes specify the failures considered to be possible at different execution parts and recovery guarantees of the proposed self-stabilizing algorithm. The set of benign faults includes both packet failures and fail-stop failures.}}
\end{figure*}

\smallskip

\noindent \textbf{Benign failures.~~}
\label{sec:benignFailures}
When the occurrence of a failure cannot cause the system execution to lose legality, \ie to leave $LE$, we refer to that failure as a benign one (Figure~\ref{fig:self-stab-SDN}).

\smallskip

\noindent \textit{Node failure.~~}
We consider \emph{fail-stop failures}, in which nodes stop taking steps. We assume that there is a way to detect these failures, say, using unreliable failure detectors~\cite{DBLP:journals/jacm/ChandraT96}. 


\smallskip

\noindent \textit{Communication failures and fairness.~~}
We consider solutions that are oriented towards time-free message-passing systems and thus they are oblivious to the time in which the packets arrive and depart. We assume that the communication channels are prone to packet failures, such as omission, duplication, reordering. However, we assume that if $p_i$ sends a message infinitely often to $p_j$, node $p_j$ receives that message infinitely often. We refer to the latter as the \emph{fair communication} assumption. For example, the proposed algorithm sends infinitely often $\mathsf{GOSSIP}$ messages from any processor to any other. Despite the possible loss of messages, the communication fairness assumption implies that every processor receives infinitely often $\mathsf{GOSSIP}$ messages from any non-failing processor. 

\smallskip

\noindent \textbf{Arbitrary transient faults.~~}
We consider any violation of the assumptions according to which the system was designed to operate. We refer to these violations and deviations as \emph{arbitrary transient faults} and assume that they can corrupt the system state arbitrarily (while keeping the program code intact). The occurrence of an arbitrary transient fault is rare. Thus, our model assumes that the last arbitrary transient fault occurs before the system execution starts~\cite{DBLP:books/mit/Dolev2000}. Moreover, it leaves the system to start in an arbitrary state.

\smallskip

\noindent \textbf{Dijkstra's self-stabilization criterion}
\label{sec:Dijkstra}
An algorithm is \textit{self-stabilizing} with respect to the task of $LE$, when every (unbounded) execution $R$ of the algorithm reaches within a finite period a suffix $R_{legal} \in LE$ that is legal. That is, Dijkstra~\cite{DBLP:journals/cacm/Dijkstra74} requires that $\forall R:\exists R': R=R' \circ R_{legal} \land R_{legal} \in LE \land |R'| \in \N$, where the operator $\circ$ denotes that $R=R' \circ R''$ concatenates $R'$ with $R''$.

%
%

\subsection{Complexity Measures}
\label{sec:timeComplexity}
The main complexity measure of self-stabilizing algorithms, called \emph{stabilization time}, is the time it takes the system to recover after the occurrence of the last transient fault. 


\smallskip
\noindent
\textbf{Message round-trips and iterations of self-stabilizing algorithms.~~}
\label{sec:messageRoundtrips}
The correctness proof depends on the nodes' ability to exchange messages during the periods of recovery from transient faults. The proposed solution considers communications that follow the pattern of request-reply, \ie $\mathsf{MSG}$ and $\mathsf{MSGack}$ messages, as well as $\mathsf{GOSSIP}$ messages for which the algorithm does not send replies. The definitions of our complexity measures use the notion of a message round-trip for the cases of request-reply messages and the term algorithm iteration.

We give a detailed definition of \emph{round-trips} as follows. Let $p_i \in \sP$ and $p_j \in \sP \setminus \{p_i\}$. Suppose that immediately after system state $c$, node $p_i$ sends a message $m$ to $p_j$, for which $p_i$ awaits a reply. At system state $c'$, that follows $c$, node $p_j$ receives message $m$ and sends a reply message $r_m$ to $p_i$. Then, at system state $c''$, that follows $c'$, node $p_i$ receives $p_j$'s response, $r_m$. In this case, we say that $p_i$ has completed with $p_j$ a round-trip of message $m$. 

It is well-known that self-stabilizing algorithms cannot terminate their execution and stop sending messages~\cite[Chapter 2.3]{DBLP:books/mit/Dolev2000}. Moreover, their code includes a do forever loop. Thus, we define a \emph{complete iteration} of a self-stabilizing algorithm. Let $N_i$ be the set of nodes with whom $p_i$ completes a message round trip infinitely often in execution $R$. Moreover, assume that node $p_i$ sends a gossip message infinitely often to $p_j \in \sP \setminus \{p_i\}$ (regardless of the message payload). Suppose that immediately after the state $c_{begin}$, node $p_i$ takes a step that includes the execution of the first line of the do forever loop, and immediately after system state $c_{end}$, it holds that: (i) $p_i$ has completed the iteration it has started immediately after $c_{begin}$ (regardless of whether it enters branches), (ii) every request-reply message $m$ that $p_i$ has sent to any node $p_j \in \sP$ during the iteration (that has started immediately after $c_{begin}$) has completed its round trip, and (iii) it includes the arrival of at least one gossip message from $p_i$ to any non-failing $p_j \in \sP \setminus \{p_i\}$. In this case, we say that $p_i$'s complete iteration (with round-trips) starts at $c_{begin}$ and ends at $c_{end}$.

\smallskip
\noindent
\textbf{Cost measures: asynchronous cycles and the happened-before relation.~~}
\label{ss:asynchronousCycles}
We say that a system execution is \emph{fair} when every step that is applicable infinitely often is executed infinitely often and fair communication is kept. Since asynchronous systems do not consider the notion of time, we use the term (asynchronous) cycles as an alternative way to measure the period between two system states in a fair execution. The first (asynchronous) cycle (with round-trips) of a fair execution $R=R' \circ R''$ is the shortest prefix $R'$ of $R$, such that each non-failing node executes at least one complete iteration in $R'$. The second cycle in execution $R$ is the first cycle in execution $R''$, and so on. 

\begin{remark}
	\label{ss:first asynchronous cycles}
	For the sake of simple presentation of the correctness proof, when considering fair executions, we assume that any message that arrives in $R$ without being transmitted in $R$ does so within $\bigO(1)$ asynchronous cycles in $R$. 
\end{remark}

\begin{remark}[Absence of transient faults implies no need for fairness assumptions]
	\label{ss:noFairnessIsNEeeded}
	In the absence of transient faults, no fairness assumptions are required in any practical settings. Also, the existing non-self-stabilizing solutions (Section~\ref{sec:back}) do not make any fairness assumption, but they do not consider recovery from arbitrary transient fault regardless of whether the execution eventually becomes fair or not.
\end{remark}

Lamport~\cite{DBLP:journals/cacm/Lamport78} defined the happened-before relation as the least strict partial order on events for which: (i) If steps $a, b \in R$ are taken by processor $p_i \in \sP$, $a \rightarrow b$ if $a$ appears in $R$ before $b$. (ii) If step $a$ includes sending a message $m$ that step $b$ receives, then $a \rightarrow b$. Using the happened-before definition, one can create a directed acyclic (possibly infinite) graph $G_R:(V_R,E_R)$, where the set of nodes, $V_R$, represents the set of system states in $R$. Moreover, the set of edges, $E_R$, is given by the happened-before relation. In this paper, we assume that the weight of an edge that is due to cases (i) and (ii) are zero and one, respectively. When there is no guarantee that execution $R$ is fair, we consider the weight of the heaviest directed path between two system state $c,c' \in R$ as the cost measure between $c$ and $c'$.      

\subsection{External building-blocks: self-stabilizing unreliable failure detectors}
\label{sec:ext}
The concepts of failure patterns and failure detectors have been introduced in~\cite{DBLP:journals/jacm/ChandraT96}. The failure detector $\Theta$ was introduced in~\cite{DBLP:conf/wdag/AguileraTD99}, and the failure detector $\mathit{HB}$ (heartbeat) has been introduced in~\cite{DBLP:journals/siamcomp/AguileraCT00}. A pedagogical presentation of these failure detectors is given in~\cite{DBLP:books/sp/Raynal18}. 

Any execution $R:=(c[0],a[0],c[1],a[1],\ldots)$ can have any number of failures during its run. $R$'s failure pattern is a function $F:\mathbb{Z}^+ \rightarrow 2^\sP$, where $\mathbb{Z}^+$ refers to an index of a system state in $R$, which in some sense represents (progress over) time, and $2^\sP$ is the power-set of $\sP$, which represent the set of failing nodes in a given system state. $F(\tau)$ denotes the set of failing nodes in system state $c_{\tau}\in R$. Since we consider fail-stop failures, $F(\tau) \subseteq F(\tau + 1)$ holds for any $\tau \in \mathbb{Z}^+$. Denote by $\mathit{Faulty}(F)\subseteq \sP$ the set of nodes that eventually fail-stop in the (unbounded) execution $R$, which has the failure pattern $F$. Moreover, $\mathit{Correct}(F)=\sP \setminus \mathit{Faulty}(F)$.

We assume the availability of self-stabilizing $\Theta$ failure detectors~\cite{DBLP:journals/siamcomp/AguileraCT00}, which offer local access to $\mathit{trusted}$, which is a set that satisfies the $\Theta$-accuracy and $\Theta$-liveness properties. Let $\mathit{trusted}^\tau_i$ denote $p_i$'s value of $\mathit{trusted}$ at time $\tau$. $\Theta$-accuracy is specified as $\forall p_i \in \sP: \forall \tau \in \mathbb{Z}^+:(\mathit{trusted}^\tau_i\cap \mathit{Correct}(F))\neq \emptyset$, \ie at any time, $\mathit{trusted}_i$ includes at least one non-faulty node, which may change over time. $\Theta$-liveness is specified as $\exists \tau \in \mathbb{N}: \forall \tau' \geq \tau : \forall p_i \in \mathit{Correct}(F):\mathit{trusted}^{\tau'}_i\subseteq \mathit{Correct}(F)$, \ie eventually $\mathit{trusted}_i$ includes only non-faulty nodes.  A self-stabilizing $\Theta$-failure detector appears in~\cite{DBLP:conf/netys/BlanchardDBD14}.

We also assume the availability of a class $\mathit{HB}$ (heartbeat) self-stabilizing failure detector~\cite{DBLP:journals/siamcomp/AguileraCT00}, which has the $\mathit{HB}$-completeness and $\mathit{HB}$-liveness properties. Let $\mathit{HB}_i^\tau[j]$ be $p_i$'s value of the $j$-th entry in the array $\mathit{HB}$ at time $\tau$. $\mathit{HB}$-completeness is specified as $\forall p_i \in \mathit{Correct}(F), \forall p_j \in \mathit{Faulty}(F): \exists K: \forall \tau \in \mathbb{N}: \mathit{HB}_i^\tau[j] < K$, \ie any faulty node is eventually suspected by every non-failing  node. $\mathit{HB}$-liveness is specified as (1) $\forall p_i, p_j \in \sP: \forall \tau \in \mathbb{N}: \mathit{HB}_i^\tau[j] \leq \mathit{HB}_i^{\tau+1}[j]$, and (2) $\forall p_i,p_j \in \mathit{Correct}(F): \forall K: \exists \tau \in \mathbb{Z}^+:\mathit{HB}_i^\tau[j] > K$. In other words, there is a time after which only the faulty nodes are suspected. The implementation of the $\mathit{HB}$ failure detector that appears in~\cite{DBLP:conf/wdag/AguileraCT97} and~\cite[Chapter 3.5]{DBLP:books/sp/Raynal18} uses unbounded counters. A self-stabilizing variation of this mechanism can simply let $p_i \in \sP$ to send $\mathsf{HEARTBEAT}(\mathit{HB}_i[i], \mathit{HB}_i[j])$ messages to all $p_j \in \sP$ periodically while incrementing the value of $\mathit{HB}_i[i]$. Once $p_j$ receives a heartbeat message from $p_i$, it updates the $i$-th and the $j$-th entries in $\mathit{HB}_j$, \ie it takes the maximum of the locally stored and received entries. Moreover, once any entry reaches the value of the maximum integer, $\mathit{MAXINT}$, a global reset procedure is used (see Section~\ref{sec:bounded}).       

\begin{remark}
	\label{ss:FD asynchronous cycles}
	For the sake of simple presentation of the correctness proof, during fair executions, we assume that $c_{\tau}\in R$ is reached within $\bigO(1)$ asynchronous cycles, such that $\forall_{p_i \in \mathit{Correct}(F)}:\mathit{trusted}^{\tau}_i\subseteq \mathit{Correct}(F)$ and for a given $K$, $\forall p_i,p_j \in \mathit{Correct}(F): \mathit{HB}_i^\tau[j] > K$, where $\tau \in \mathbb{Z}^+$ is determined by the $\Theta$- and $\mathit{HB}$-liveness properties. 
\end{remark}

\section{Non-self-stabilizing URB with and without Failure Detectors}
\label{sec:back}
For the completeness' sake, we briefly review existing URB solutions. The following algorithms are from \cite{DBLP:journals/siamcomp/AguileraCT00,DBLP:conf/wdag/AguileraTD99}. We follow here their description as give in~\cite{DBLP:books/sp/Raynal18} by starting from the simplest model before considering more advanced ones. 

\begin{algorithm*}[h!]
	\begin{\algSize}
		
		\smallskip
		
		\textbf{operation} $\mathsf{urbBroadcast}(m)$ \textbf{do} \textbf{send} $\mathsf{MSG}(m)$ \textbf{to} $p_i$\label{ln:urbRealChannel:broadcast}\;
		
		\smallskip
		
		\textbf{upon} $\mathsf{MSG}(m)$ \textbf{arrival from} $p_k$ \Begin{ 
			\If{first reception of $m$\label{ln:urbRealChannel:arrival}}{
				\{\lForEach{$p_j \in \sP \setminus \{p_i,p_k\}$}{\textbf{send} $\mathsf{MSG}(m)$ \textbf{to} $p_j$\label{ln:urbRealChannel:forward}\}; $\mathsf{urbDeliver}(m)$\label{ln:urbRealChannel:deliver}}
			}
		}
		
		\caption{\label{alg:urbRealChannel}URB in the presence of reliable communications;  code for $p_i\in\sP$}	
		
	\end{\algSize}
	
\end{algorithm*}

In the absence of communication and node failures, one can implement the $\mathsf{urbBroadcast}(m)$ operation by running \{\textnormal{foreach} $p_j \in \sP$ \textnormal{send} $\mathsf{MSG}(m)$ \textnormal{to} $p_j$\} and calling  $\mathsf{urbDeliver}(m)$ upon $p_j$'s reception of $m$. 
Algorithm~\ref{alg:urbRealChannel} considers a model in which nodes can fail-stop without the possibility to detect it, but with reliable communications. Node $p_i$ broadcasts message $m$ by sending $\mathsf{MSG}(m)$ to itself (line~\ref{ln:urbRealChannel:broadcast}). Upon the message arrival (line~\ref{ln:urbRealChannel:arrival}), the receiver ignores the message if it got it before. This is possible due to the requirement of unique message identities (Definition~\ref{sec:spec}). If it is the first reception, the receiver propagates $\mathsf{MSG}(m)$ to all other nodes (except itself and the sender) before calling  $\mathsf{urbDeliver}(m)$ (lines~\ref{ln:urbRealChannel:forward} to~\ref{ln:urbRealChannel:deliver}).


Algorithm~\ref{alg:urbUnRealChannel} considers a system in which at most $t < n/2$ nodes may crash without the possibility for detection as well as unreliable communications. Node $p_i$ broadcasts message $m$ by sending $\mathsf{MSG}(m)$ to itself (line~\ref{ln:urbRealChannel:broadcast2}) while assuming it has a reliable channel to itself). Upon the reception of $\mathsf{MSG}(m)$ for the first time (line~\ref{ln:notfirst2}), $p_i$ creates the set $\mathit{recBy}[m]=\{i,k\}$ to contain the identities of nodes that receive $\mathsf{MSG}(m)$, before activating the $\mathit{Diffuse}(m)$ task. In case this is not $\mathsf{MSG}(m)$'s first arrival (line~\ref{ln:allocate2}),  $p_i$ merely adds the sender identity, $k$, to $\mathit{recBy}[m]$. The task $\mathit{Diffuse}(m)$ is responsible for transmitting (and retransmitting) $\mathsf{MSG}(m)$ to at least a majority of the nodes before URB-delivering $m$ (lines~\ref{ln:send2} to~\ref{ln:urbDeliver2}).

\begin{algorithm*}[h!]
	\begin{\algSize}
		
		\smallskip
		
		\textbf{operation} $\mathsf{urbBroadcast}(m)$ \textbf{do} \textbf{send} $\mathsf{MSG}(m)$ \textbf{to} $p_i$\label{ln:urbRealChannel:broadcast2}\;
		
		\smallskip
		
		\textbf{upon} $\mathsf{MSG}(m)$ \textbf{arrival from} $p_k$ \Begin{ 
			\lIf{not the first reception of $m$}{$\mathit{recBy}[m] \gets \mathit{recBy}[m] \cup \{k\}$\label{ln:notfirst2}}
			\lElse{\textbf{allocate} $\mathit{recBy}[m]$; $\mathit{recBy}[m] \gets \{i,k\}$;
				\textbf{activate} $\mathit{Diffuse}(m)$ \textbf{task}\label{ln:allocate2}}
		}
		
		\smallskip
		
		\textbf{do forever} \Begin{

			\ForEach{\emph{\textbf{active}} $\mathit{Diffuse}(m)$ \textbf{task}}{
				\lForEach{$p_j \in \sP : j \notin \mathit{recBy}[m]$}{\textbf{send} $\mathsf{MSG}(j, seq)$ \textbf{to} $p_j$\label{ln:send2}}
				\lIf{$|\mathit{recBy}[m]| \geq t + 1) \land (p_i \text{ has not yet URB-delivered } m)$}{$\mathsf{urbDeliver}(m)$\label{ln:urbDeliver2}}
			}
			
		}
		
		\caption{\label{alg:urbUnRealChannel}URB in the presence of $t<n/2$ undetectable node failures; $p_i$'s code}	
		
	\end{\algSize}
	
\end{algorithm*}

Note that the task $\mathit{Diffuse}(m)$ never stops transmitting messages. Using $\Theta$ failure detectors (Section~\ref{sec:spec}), Algorithm~\ref{alg:urbRealChannelMore} avoids such an infinite number of retransmissions by enriching Algorithm~\ref{alg:urbUnRealChannel} as follows. (i) The URB-delivery condition, $\mathit{trusted} \subseteq \mathit{recBy}[m]$, of Algorithm~\ref{alg:urbRealChannelMore}'s line~\ref{ln:urbDeliver3} substitutes the condition, $|\mathit{recBy}[m]| \geq t + 1)$, of Algorithm~\ref{alg:urbUnRealChannel}'s line~\ref{ln:urbDeliver2}. (ii) Upon the reception of a $\mathsf{MSG}(m)$ message, $p_i$ acknowledges the reception via a $\mathsf{MSGack}(m)$. Moreover, when $p_i$ receives $\mathsf{MSGack}(m)$ from $p_k$, it marks the fact that $p_k$ received $m$ by adding $k$ to $\mathit{recBy}[m]$. (iii) Node $p_i$ can eventually avoid sending $\mathsf{MSG}(m)$ messages to a faulty processor $p_j$ in the following manner. Processor $p_i$ repeatedly transmits $\mathsf{MSG}(m)$ to $p_j$ as long as $p_j$ is trusted and $j \notin \mathit{recBy}[m]$ (line~\ref{ln:condSend3}). Note that, eventually, either $p_j$ will receive $\mathsf{MSG}(m)$ and acknowledge it to $p_i$, or in case $p_j$ is faulty, $j \notin \mathit{trusted}_i$ due to the $\Theta$-completeness property. Moreover, due to the strong $\Theta$-accuracy, $j \notin \mathit{trusted}_i$ cannot hold before $p_j$ fails (if it is faulty).

\begin{algorithm*}[h!]
	\begin{\algSize}
		
		\smallskip
		
		\textbf{operation} $\mathsf{urbBroadcast}(m)$ \textbf{do} \textbf{send} $\mathsf{MSG}(m)$ \textbf{to} $p_i$\;
		
		\smallskip
		
		\textbf{upon} $\mathsf{MSG}(m)$ \textbf{arrival from} $p_k$ \Begin{ 
			\lIf{not the first reception of $m$}{$\mathit{recBy}[m] \gets \mathit{recBy}[m] \cup \{k\}$}
			\lElse{\textbf{allocate} $\mathit{recBy}[m]$; $\mathit{recBy}[m] \gets \{i,k\}$;
				\textbf{activate} $\mathit{Diffuse}(m)$ \textbf{task}}
			\textbf{send} $\mathsf{MSGack}(m)$ \textbf{to} $p_k$\;
		}
		
		\smallskip
		
		\textbf{upon} $\mathsf{MSGack}(m)$ \textbf{arrival from} $p_k$ \textbf{do} \{$\mathit{recBy}[m] \gets \mathit{recBy}[m] \cup \{k\}$\}
		
		\smallskip
		
		\textbf{do forever} \Begin{

			\ForEach{\emph{\textbf{active}} $\mathit{Diffuse}(m)$ \emph{\textbf{task}}}{
				\lForEach{$j \in \mathit{trusted} \setminus \mathit{recBy}[m]$}{\textbf{send} $\mathsf{MSG}(m)$ \textbf{to} $p_j$\label{ln:condSend3}}
				\lIf{$\mathit{trusted} \subseteq \mathit{recBy}[m] \land (p_i \text{ has not yet URB-delivered } m)$}{$\mathsf{urbDeliver}(m)$\label{ln:urbDeliver3}}
			}
			
		}
		
		\caption{\label{alg:urbRealChannelMore}Quiescent URB using  $\Theta$-failure detectors; code for $p_i \in \sP$}	
		
	\end{\algSize}
	
\end{algorithm*}


To the end of allowing the implementation of a quiescent URB solution and the unreliable failure detectors that it relies on, the underlying system needs to satisfy synchrony assumptions that can be captured by the combined use of the $\Theta$- and $\mathit{HB}$-failure detectors~\cite[Chapter 3.5]{DBLP:books/sp/Raynal18}. Algorithm~\ref{alg:urbRealChannelSo} differs from Algorithm~\ref{alg:urbRealChannelMore} only in the $\mathit{Diffuse}(m)$ task (line~\ref{ln:snedOnly}). Specifically, $p_i$ transmits $\mathsf{MSG}(m)$ to $p_j$ only when $j \in \mathit{recBy}[m]$ (because from $p_i$'s perceptive, $p_j$ has not yet received $\mathsf{MSG}(m)$) and $\mathit{HB}[j]$ has increased since the previous iteration (because from $p_i$’s perspective, $p_j$ is not failing). Algorithm~\ref{alg:urbRealChannelSo} is the basis for our proposal (Section~\ref{sec:basic}).

\begin{algorithm*}[h!]
	\begin{\algSize}
		
		\smallskip
		
		\textbf{operation} $\mathsf{urbBroadcast}(m)$ \textbf{do} \textbf{send} $\mathsf{MSG}(m)$ \textbf{to} $p_i$\;
		
		\smallskip
		
		\textbf{upon} $\mathsf{MSG}(m)$ \textbf{arrival from} $p_k$ \Begin{ 
			\lIf{not the first reception of $m$}{$\mathit{recBy}[m] \gets \mathit{recBy}[m] \cup \{k\}$}
			\lElse{\textbf{allocate} $\mathit{recBy}[m]$; $\mathit{recBy}[m] \gets \{i,k\}$;
				\textbf{activate} $\mathit{Diffuse}(m, [\text{-}1,\ldots,\text{-}1])$ \textbf{task}}
			\textbf{send} $\mathsf{MSGack}(m)$ \textbf{to} $p_k$\;
		}
		
		\smallskip
		
		\textbf{upon} $\mathsf{MSGack}(m)$ \textbf{arrival from} $p_k$ \textbf{do} \{$\mathit{recBy}[m] \gets \mathit{recBy}[m] \cup \{k\}$\}
		
		\smallskip
		
		\textbf{do forever} \Begin{

			\ForEach{\emph{\textbf{active}} $\mathit{Diffuse}(m, \mathit{prevHB})$ \emph{\textbf{task}}}{
				\textbf{let} $\mathit{curHB}:=\mathit{HB}$\;
				\ForEach{$j \in \mathit{trusted} \setminus \mathit{recBy}[m] \land \mathit{prevHB}[m][j] <\mathit{curHB}[m][j] $}{\textbf{send} $\mathsf{MSG}(m)$ \textbf{to} $p_j$\label{ln:snedOnly}}
				$\mathit{prevHB}[m] \gets \mathit{curHB}[m]$\;
				\lIf{$\mathit{trusted} \subseteq \mathit{recBy}[m] \land (p_i \text{ has not yet URB-delivered } m)$}{$\mathsf{urbDeliver}(m)$}
			}
			
		}
		
		\smallskip
		
		\caption{\label{alg:urbRealChannelSo}Quiescent URB using  $\Theta$- and $\mathit{HB}$-failure detectors; code for $p_i \in \sP$}	
		
	\end{\algSize}
	
\end{algorithm*}


\section{Unbounded Self-stabilizing Uniform Reliable Broadcast}
\label{sec:basic}
Algorithm~\ref{alg:qurb} allows $p_i \in \sP$ to $\mathsf{urbBroadcast}$ message $m$ in a way the guarantees that all non-failing nodes raise the event $\mathsf{urbDeliver}(m)$ according to the specifications (Section~\ref{sec:spec}). The review in Section~\ref{sec:back} can help the reader to understand the proposed solution. We note that the boxed code lines of Algorithm~\ref{alg:qurb} are relevant only for an extension, which we discuss in Section~\ref{sec:fifo}.

%

\smallskip
\noindent
\textbf{Local variables and their purpose (lines~\ref{ln:varCap} to~\ref{ln:varTxS}).~~}
The task specifications assume that each processor $p_i \in \sP$ can URB-broadcast unique messages. To that end, Algorithm~\ref{alg:qurb} maintains the message index number, $seq$, that it increments upon $\mathsf{urbBroadcast}$ invocations.

The processors store all the currently processed messages as records in the variable $\mathit{buffer}$. Each record includes the following fields: (i)  $msg$, which holds the URB message, (ii) $id$, which is the identifier of the node that invoked the URB-broadcast, (iii) $seq$, which is the message index number, (iv) $\mathit{delivered}$, which is a Boolean that holds $\false$ only when the message is pending delivery, (v) $\mathit{recBy}$, which is a set that includes the identifiers of nodes that have acknowledged $msg$, and (vi) $\mathit{prevHB}$, which is a value of the $\mathit{HB}$ failure detector (Section~\ref{sec:spec}) that Algorithm~\ref{alg:qurb} uses for deciding when to transmit (and re-transmit) $msg$. Our proof shows that every node store at most $n \cdot \buffCapacity$ records, where $\buffCapacity$ can be set according to the available local memory. When accessing records in $\mathit{buffer}$, we use a query-oriented notation, \eg $(\bullet,id=j,seq=s,\bullet) \in \msgSet$ considers all buffered records that their $id$ and $seq$ fields hold the values $j$ and $s$, respectively.

\begin{figure*}[t!]
	\begin{\algSize}
		\centering
		\hspace*{-0.5em}\includegraphics[clip=true,scale=0.9]{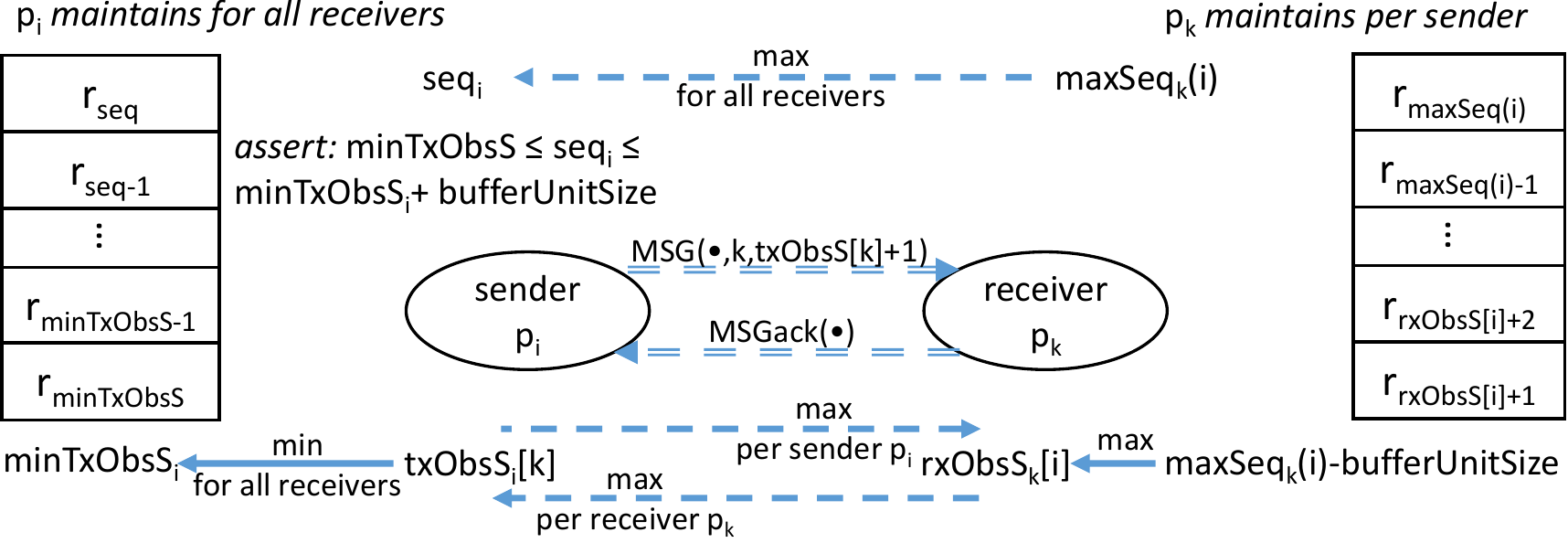}\\
		\caption{\label{fig:flowControl}{The self-stabilizing flow-control scheme between sender $p_i$ and receiver $p_k$. The arrays on the figure sides represent the portion of peers' $\msgSet$ variables that includes records $r_s$, where $s$ is a sequence number of a message sent from $p_i$ to $p_k$. The single-line arrows (dashed or not) and the text next to them represent a logical update, \eg $x \xleftarrow{\text{max}} y$ stands for $x \gets \max \{x,y\}$. The text that appears below the arrow clarify whether a single variable aggregates these update or different entries in the array store the updated values. The dashed arrows refer to updates that require communication between $p_i$ and $p_k$. The double-line arrows and the text above them depict $\mathsf{MSG}$ and $\mathsf{MSGack}$ messages.}}
	\end{\algSize}
\end{figure*}

\smallskip
\noindent
\textbf{A self-stabilizing flow-control scheme for bounding $\msgSet$.~~}
Algorithm~\ref{alg:qurb} bounds $\msgSet$ using a flow-control technique. We say a record, with sequence number $s$, is \emph{obsolete} if it had received acknowledgments from all trusted nodes and then it was URB-delivered. Moreover, since $p_i$ needs to remove obsolete records from its buffer, we also define that any record with a sequence number lower than $s$ to be also obsolete. This way, $p_k$ can keep track of all the obsolete records it has deleted using a single counter $\rxS_k[i]$, per sender $p_i$, which stores the highest sequence number of records that $p_i$ considers to be obsolete. The array $\txS_i[]$ facilitates the ability of sender ($p_i$) to control its sending flow since it can receive $\rxS_k[i]$ from $p_k$ and store it at $\txS_i[k]$. (Note that we denote variable $X$'s value at node $p_i$ by $X_i$.) The flow-control mechanism can simply defer the processing a new URB-message when $p_k$'s message sequence number minus the minimum value stored in $\txS[]$ (that arrived from a node that $p_k$ trusts) is smaller than the maximum number of records, $\buffCapacity$, that a receiver can buffer for each node.

We use Figure~\ref{fig:flowControl} to describe in detail the flow-control scheme. The receiver $p_k$ repeatedly sends to the sender $p_i$ the maximum $p_i$'s sequence number, $\maxS_k(i)$, that it stores in its buffer, see the top dashed left arrow. This allows $p_i$ to make sure that $seq_i$ is greater than any sequence number in the system that is associated with $p_i$, as we show in Theorem~\ref{thm:recovery}'s Argument (3). The buffer of $p_k$ cannot store more than $\buffCapacity$ with messages from $p_i$. Therefore, $p_k$ stores only messages that their sequence numbers are between $\maxS_k(i)$ and $\maxS_k(i)-\buffCapacity$ and reports to $p_i$ the highest sequence number, $\rxS_k[i]$, of its obsolete records that are associated with $p_i$, see the lowest dashed arrow. The latter stores this value in $\rxS_i[k]$ and makes sure it has the latest value from $p_k$ by sharing $\rxS_i[k]$ with it. The sender $p_i$ also uses $\rxS_i[k]$ for bounding $\msgSet_i$. Specifically, $\minTxObstSeq_i()$ aggregates the minimum value in $\rxS_i[k]$ for any trusted receiver $p_k$ (line~\ref{ln:minTxObstSeq}). Using $\minTxObstSeq_i()$, the sender $p_i$ can assert that $\minTxObstSeq_i() \leq seq_i \leq \minTxObstSeq_i()+\buffCapacity$ and $\msgSet_i$ includes all the records that their sequence numbers are between $\minTxObstSeq_i()$ and $seq_i$ (line~\ref{ln:bufferMinTxObstSeq}). Since, due to a transient fault, $p_i$'s might indicate the reception of acknowledgment for a message that $p_k$'s state shows that it has never received, $p_i$ repeatedly resends the message that has the sequence number $s$, such that $s=\rxS_i[k]+1$ (line\ref{ln:sendMSGForEach}), see the double line arrows between $p_i$ and $p_k$ in Figure~\ref{fig:flowControl}.


\begin{algorithm*}[t!]
	\begin{\algSize}
		
		\smallskip
		
		\noindent \textbf{global constants:}
		$\buffCapacity$\label{ln:varCap}\tcc*{max records per node in $\msgSet$}
		
		\smallskip
		
		\noindent \textbf{local variables:} (Initialization is optional in the context of self-stabilization.)\\
		$seq:=0$\label{ln:varSeq}\tcc*{message index num.}
		$\msgSet:=\emptyset$\label{ln:varBuffer}\tcc*{set of $(msg,id,seq,\mathit{delivered},\mathit{recBy},\mathit{prevHB})$ records}
		$\rxS[1..n]:=[0,\ldots,0]$\label{ln:varRxS}\tcc*{highest reciver's obsolete seq per node} 
		$\txS[1..n]:=[0,\ldots,0]$\label{ln:varTxS}\tcc*{highest sender's obsolete seq per node} 
		\fbox{$next[1..n]:=[1,\ldots,1]$}\label{ln:varNext}\tcc*{next-to-deliver message indices; one entry per sender}
		
		\smallskip
		
		
		\textbf{interface required} $\mathit{trusted}$ and $\mathit{HB}$\tcc*{see Sec.~\ref{sec:spec}} 
		
%
		
		\smallskip
		
		\textbf{macro} $\obsolete(r) := (\rxS[r.id] +1= r.seq \land \mathit{trusted} \subseteq r.\mathit{recBy} \land r.\mathit{delivered})$\label{ln:obsolete}\;
		
		\smallskip
		
		\textbf{macro} $\maxS(k) := \max (\{ s: (\bullet, id=k,seq=s,\bullet) \in \msgSet\}$\fbox{$\cup \{next[k]-1\}$}$)$\label{ln:maxS}\;
		
		\smallskip
		
		\textbf{macro} $\minTxObstSeq() := \min\{ \txS[k] :k \in \mathit{trusted}\}$\label{ln:minTxObstSeq}\;
		
		\smallskip
		
		\textbf{operation} $\mathsf{urbBroadcast}(m)$ \label{ln:oprurbBCast}\textbf{do} \{\textbf{wait}$(seq<\minTxObstSeq()+\buffCapacity)$; $seq \gets$ $seq+1$; $\mathsf{update}(m, i, seq, i)$; \label{ln:call2urbBroadcast}\} \tcc*{returns the transmission descriptor}
		
		\smallskip
		
		\textbf{procedure} $\mathsf{update}(m, j, s, k)$ \Begin{
			\lIf{$s \leq \rxS[j]$}{\Return \label{ln:updateReturn}}
			\If{$(\bullet, id=j, seq=s, \bullet) \notin \msgSet \land m \neq \bot$}{$\msgSet \gets \msgSet \cup \{(m, j, s, \false, \{j, k\}, [\text{-}1,\ldots,\text{-}1])\}$\label{ln:addBuff}\;}
			\lElse{\textbf{foreach} {$(\bullet, id=j, seq=s, \bullet, recBy=r, \bullet) \in \msgSet$} \textbf{do} {$r \gets r \cup \{j, k\}$\label{ln:addJack}}}
		}
		
		\smallskip
		
		\textbf{do forever} \Begin{\label{ln:doForever}
			
			\lIf{$(\exists r,r'\in \msgSet$$: r.msg$$= \bot \lor (r\neq r' \land ((r.id,r.seq)$$=(r'.id,r'.seq))))$}{$\msgSet \gets \emptyset$\label{ln:noDubOrBot}}		
			
			\lIf{$\neg ((mS \leq seq \leq mS\mathrm{+}\buffCapacity)$$\land (\{mS+1, \ldots ,seq\}$$\subseteq \{s : (\bullet,id=i,seq=s,\bullet)$$\in \msgSet\})$ \emph{\textbf{where}} $mS := \minTxObstSeq()$}{$\txS[] \gets [seq, \ldots, seq]$\label{ln:bufferMinTxObstSeq}}
			
			\lForEach{$p_k \in \sP$}{$(\rxS[k]\fbox{,next[k])} \gets (\max \{\rxS[k],\maxS(k) - \buffCapacity \},$ \fbox{$\max\{ next[k],\rxS[k]\ems{+1}\}$})\label{ln:rxSmaxSBuffCapacity}}
			
			\lWhile{$\exists r \in \msgSet: \obsolete(r)$}{$\rxS[r.id] \gets \rxS[r.id]+1$\label{ln:recBytrustedBuffer}}
			
			$\msgSet \gets \{ (\bullet,id=i,seq=s,\bullet) \in \msgSet: \minTxObstSeq() < s \} \cup$ $\{ (\bullet,id=k,seq=s,\bullet) \in \msgSet: p_k \in \sP \land ((\rxS[k] < s \land \maxS(k) - \buffCapacity \leq s)) \}$\label{ln:buffLimit}\;
			
			
			\ForEach{$(msg=m,id=j,seq=s,\mathit{delivered}=d,\mathit{recBy}=r,\mathit{prevHB}=e) \in \msgSet$\label{ln:forEachMSG}}{

				\If{$(\mathit{trusted} \subseteq \mathit{r})\land(\neg d)$\fbox{$\land s=next[k]$}\label{ln:trustedSubsetC}}{$\mathsf{urbDeliver}(m); d \gets \true;$\fbox{$next[k]\gets next[k]\mathrm{+}1$}\label{ln:trustedSubsetR}} 

				\textbf{let} $u := \mathit{HB}$\label{ln:deliveredTrue}\;
				
				\ForEach{$p_k \in \sP : (k \notin r \lor (i=j \land s=\txS[k]+1)) \land (e[k] < u[k])$\label{ln:sendMSGForEach}}{$e[k] \gets u[k]$; \textbf{send} $\mathsf{MSG}(m,j,s)$ \textbf{to} $p_k$\tcc*{piggyback lines~\ref{ln:sendMSG} and~\ref{ln:URBsendGossip}}\label{ln:sendMSG}}

				
			}
			
			\lForEach{$p_k \in \sP$}{$\mathbf{send}~ \mathsf{GOSSIP}(\maxS(k), \rxS[k], \txS[k] )~ \mathbf{to}~ p_k$\label{ln:URBsendGossip}}
		}
		
		\smallskip
		
		\textbf{upon} $\mathsf{MSG}(m, j, s)$ \textbf{arrival from} $p_k$ \textbf{do} {\label{ln:MSGarrive}\{$\mathsf{update}(m, j, s, k)$; \textbf{send} $\mathsf{MSGack}(j, s)$ \textbf{to} $p_k$;\label{ln:sendMSGack}\}
		}
		
		\smallskip
		
		\textbf{upon} $\mathsf{MSGack}(j, s)$ \textbf{arrival from} $p_k$ \textbf{do} \label{ln:MSGackarrive}\{$\mathsf{update}(\bot, j, s, k)$;\label{ln:arrMSGack}\}
		
		\smallskip
		
		\textbf{upon} $\mathsf{GOSSIP}(\textit{seqJ},\txSJ,\rxSJ)$ \textbf{arrival} \textbf{from} $p_j$ \textbf{do} {\{$(seq,\txS[j]$, $\rxS[j])\gets (\max \{\textit{seqJ}, seq\}, \max \{\txSJ, \txS[j]\}, \max \{\rxSJ, \rxS[j]\})$;\}\label{ln:urbGOSSIPupdate}}
		
		\smallskip
		
		\caption{\label{alg:qurb}Self-stabilizing quiescent uniform reliable broadcast; code for $p_i \in \sP$}	
		
	\end{\algSize}
	
\end{algorithm*}

\smallskip
\noindent
\textbf{A detailed description of Algorithm~\ref{alg:qurb}.~~}
Upon the invocation of the $\mathsf{urbBroadcast}(m)$ operation, Algorithm~\ref{alg:qurb} allows node $p_i$ to process $m$ without blocking as long as the flow-control mechanism can guarantee the available space at all trusted receivers (line~\ref{ln:call2urbBroadcast}). Such processing is done by creating a unique operation index, $seq$, and calling $\mathsf{update}()$.

The procedure $\mathsf{update}(m, j, s, k)$ receives a message, $m$, a unique message identifier, which is the pair $(j, s)$ that includes the sender identifier ($j$) and the sequence number ($s$), and the identifier of the forwarding processor, $p_k$. The procedure considers first the case in which $\msgSet_i$ does not include a record with the identifier $(j, s)$. In this case, $p_i$ adds to $\msgSet_i$ the record $(m, j, s, \false, \{j, k\}, [\text{-}1,\ldots,\text{-}1])$ (line~\ref{ln:addBuff}), which stands for the message itself and its unique identifier, as well as stating that it was not yet been delivered but that the identifiers of the sending ($j$) and forwarding ($k$) processors appear in $\mathit{recBy}$. Moreover, the record holds a vector that is smaller than any value of the $\mathit{HB}$ failure detector. For the case in which $\msgSet_i$ already includes a record with the identifier $(j, s)$, $p_i$ makes sure that $\mathit{recBy}$ includes the identifiers of the sending and forwarding nodes (line~\ref{ln:addJack}).      

Algorithm~\ref{alg:qurb} includes a do forever loop (lines~\ref{ln:doForever} to~\ref{ln:URBsendGossip}) that: (i) removes stale information (lines~\ref{ln:noDubOrBot} to~\ref{ln:buffLimit}), (ii) processes URB messages (lines~\ref{ln:forEachMSG} to~\ref{ln:sendMSG}) and (iii) gossips information that is needed for flow-control and recovery from arbitrary transient faults (line~\ref{ln:URBsendGossip}). 

\smallskip

\textbf{(i)} The removal of stale information includes the emptying the buffer whenever there are records for which the $msg$ field is $\bot$ or when there are two records with the same message identifier (line~\ref{ln:noDubOrBot}). Lines~\ref{ln:bufferMinTxObstSeq} to~\ref{ln:recBytrustedBuffer} implement recovery strategies that facilitate the bounds on the buffer size. Algorithm~\ref{alg:qurb} tests for the case in which, due to an arbitrary transient fault, the sender does not store all of its messages such that their sequence number is between $mS\text{+}1$ and $seq$ (line~\ref{ln:bufferMinTxObstSeq}), where $mS := \minTxObstSeq()$ is the smallest obsolete sequence number that $p_i$ had received from a trusted receiver. The recovery from such transient violations is done by allowing the sender to send $\buffCapacity$ URB messages without considering the space available on the receiver-side. Similarly, on the receiver-side, Algorithm~\ref{alg:qurb} makes sure that the gap between the largest obsolete record, $\rxS[k]$ (of $p_k$'s messages) and the largest buffered sequence number, $\maxS(k)$, is not larger than $\buffCapacity$  (line~\ref{ln:rxSmaxSBuffCapacity}). Algorithm~\ref{alg:qurb} updates the receiver-side counter that stores the highest obsolete message number per sender (line~\ref{ln:recBytrustedBuffer}). To the end of bounding the memory use, $p_i$ keeps in $\msgSet_i$ messages that it has sent but for which it has not yet received an indication from all trusted receivers that they consider this message to be obsolete. It also keeps all non-obsolete messages (regardless of their sender).


\smallskip

\textbf{(ii)} Node $p_i$ processes records by testing the field $\mathit{recBy}$ of any not delivered message (line~\ref{ln:forEachMSG}). The message is delivered when $\mathit{recBy}$ encodes an acknowledgment from every trusted node (line~\ref{ln:trustedSubsetR}). Processor $p_i$ then marks the record as a delivered one and samples the $\mathit{HB}$ failure detector (line~\ref{ln:deliveredTrue}). This sample is used to decide when a transmission (or retransmission) of a URB message is needed (line~\ref{ln:sendMSGForEach}) in case an acknowledgment is missing or because the message sequence number is greater by one than the largest obsolete message number known to the sender. These messages are received and processed in line~\ref{ln:sendMSGack}, which includes acknowledging the message arrival. These acknowledgments are processed in line~\ref{ln:arrMSGack}.

\smallskip

\textbf{(iii)} At the end of the do-forever loop, $p_i$ gossips to every $p_k$ control information about the maximum $seq$ value that $p_i$ stores in a $p_k$ record as well as $p_k$'s obsolete records (lines~\ref{ln:URBsendGossip} and~\ref{ln:urbGOSSIPupdate}). The former value allows $p_k$ to maintain the correctness invariant, \ie $seq_k$ is not smaller than any other $seq$ value in the system that is associated with $p_k$. The latter value allows $p_k$ to control the flow of URB broadcasts according to the available space in $\msgSet_i$.       

\section{Correctness} 
\label{sec:qurbCorr}
This section brings the correctness proof of Algorithm~\ref{alg:qurb}. Theorem~\ref{thm:recovery} demonstrates recovery after the occurrence of the last arbitrary transient fault. Theorem~\ref{thm:closure} demonstrates that Algorithm~\ref{alg:qurb} satisfies the task specifications (Section~\ref{sec:spec}).

\subsection{Needed definitions} 
Definition~\ref{def:safeConfig} presents the necessary conditions for demonstrating that Algorithm~\ref{alg:qurb} brings the system to a legal execution (Theorem~\ref{thm:recovery}).   


\begin{definition}[Algorithm~\ref{alg:qurb}'s consistent sequence and buffer values]
	\label{def:safeConfig}
	Let $c$ be a system state and $p_i \in\sP$ a non-faulty processor. Suppose that (i)  $(\nexists r,r'\in \msgSet: r.msg = \bot \lor (r\neq r' \land ((r.id,r.seq)=(r'.id,r'.seq))))$, $((mS \leq seq_i \leq mS + \buffCapacity) \land (mS\text{+}1, \ldots ,seq_i\} \subseteq \{s : (\bullet,id=i,seq=s,\bullet) \in \msgSet_i\})$, $\forall p_k \in \sP: (\maxS_i(k)-\rxS_i[k]) \leq \buffCapacity$, $\nexists r \in \msgSet_i : \obsolete(r)$, $\forall (\bullet,id=i,seq=s,\bullet) \in \msgSet_i: mS < s$, $\forall (\bullet,id=k,seq=s,\bullet) \in \msgSet_i: p_k \in \sP \land \rxS_i[k] < s \land \maxS_i(k) \leq (s+\buffCapacity)$, where $mS:=\minTxObstSeq_i()$. 
	%
	%
	Moreover, (ii) $seq_i$ is greater than or equal to any $p_i$'s sequence values in the variables and fields related to $seq$ (including $p_i$'s records in $\msgSet_k$, where $p_k \in \sP$ is non-failing, and incoming messages to $p_k$) and $\forall p_j \in \sP:sMj\leq \rxS_j[i]$, where $sMj$ is either $\txS_i[j]$ or the value of the fields $\txSJ$ and $\rxSJ$ in a $\mathsf{GOSSIP}(\bullet,\txSJ,\bullet)$ message in transit from $p_j$ to $p_i$, and respectively, $\mathsf{GOSSIP}(\bullet,\rxSJ)$ message in transit from $p_i$ to $p_j$. Also, (iii)  $\forall k \in \mathit{trusted}_i : | \{ (\bullet,id=i,\bullet) \in \msgSet_k \}|\leq \buffCapacity \land seq_i\leq\minTxObstSeq_i()+\buffCapacity$. In this case, we say that $p_i$'s values in the variables and fields related to $seq$'s sequence values and $\msgSet$ are consistent in $c$.
\end{definition}

We note that not every execution that starts from a system state that satisfies Definition~\ref{def:safeConfig} is a legal execution. For example, consider a system with $\sP=\{p_i,p_j\}$ and an execution $R$, such that in its starting state it holds that $\msgSet_i=\{(m,i,1,\bullet)\}$ and $\msgSet_j=\{(m',i,1,\bullet)\}$, such that $m\neq m'$. The delivery of $m$ and $m'$ violates Definition~\ref{def:URB}'s validity requirement because no legal execution has $R$ as a suffix. Theorem~\ref{thm:closure} circumvents this difficulty using Definition~\ref{def:complete}. Thus, definitions~\ref{def:safeConfig} and~\ref{def:complete} provide the necessary and sufficient conditions for demonstrating self-stabilization. 

\begin{definition}[Complete execution with respect to $\mathsf{urbBroadcast}$ invocations]
	\label{def:complete}
	Let $R$ be an execution of Algorithm~\ref{alg:qurb}. Let $c, c'' \in R$ denote the starting system states of $R$, and respectively, $R''$, for some suffix $R''$ of $R$. We say that message $m$ is \emph{completely delivered} in $c$ if (i) the communication channels do not include $\mathsf{MSG}(msg=m,\bullet)$ messages (or $\mathsf{MSGack}$ messages with a message identifier $(id,seq)$ that refers to $m$), and (ii) for any non-failing $p_j \in\sP$ and $r=(msg=m,\bullet) \in \msgSet_j$, it holds that $r.\mathit{delivered}=\true$ and for any non-failing $p_k \in\sP$, we have $k \in r.\mathit{recBy}$. 
	%
	%
	Suppose that $R=R' \circ R''$ has a suffix $R''$, such that for any $\mathsf{urbBroadcast}$ message $m$ that is not completely delivered in $c''$, it holds that $m$ either does not appear in $c$ (say, due to an $\mathsf{urbBroadcast}(m)$ invocation during the prefix $R'$) or it is completely delivered in $c$.
	%
	%
	In this case, we say that $R''$ is complete with respect to $R$'s invocations of $\mathsf{urbBroadcast}$ messages. When the prefix $R'$ is empty, we say that $R$ is complete with respect to $\mathsf{urbBroadcast}$ invocations. 
\end{definition}

Theorems~\ref{thm:recovery} and~\ref{thm:closure} consider Definition~\ref{def:diffuse}. For the sake of simple presentation, Definition~\ref{def:diffuse} relies on the specifications' assumption (Section~\ref{sec:spec}) that every broadcasted message is unique (even without an explicit assignment of the message identifier $(id,seq)$ to $m$). 

\begin{definition}[The $\mathit{diffuse}()$ predicate]
	\label{def:diffuse}
	Let $p_i\in \sP$ and $c\in R$ be a system state. The predicate $\mathit{diffuse}_i(m)$ holds in $c$ if, and only if, $\exists (msg=m,\bullet,\mathit{delivered}=\false,\bullet) \in \msgSet_i$.
\end{definition}

\subsection{Basic facts} 
Both theorems~\ref{thm:recovery} and~\ref{thm:closure} use Lemma~\ref{thm:basic}. 


\begin{lemma}
	\label{thm:basic}
	Let $R$ be an execution of Algorithm~\ref{alg:qurb} and $p_i, p_j \in \sP$ be two non-failing processors. Suppose that in every system state $c \in R$,  $\mathit{diffuse}_i(m)$ holds, such that $(msg\neq \bot,\bullet, recBy=r, \bullet) \in \msgSet_i$ holds, but $j \in r$ does not. (i) Processor $p_i$ sends, infinitely often, the message $\mathsf{MSG}(m,j,s)$ to $p_j$ and $p_j$ acknowledges, infinitely often, via the message $\mathsf{MSGack}(j,s)$ to $p_i$. (ii) The reception of any of these acknowledgments guarantees that $j \in r$ holds within cost measure of $2$ (Section~\ref{ss:asynchronousCycles}). (iii) Suppose that $R$ is fair, then invariants (i) and (ii) occur within  $\bigO(1)$ asynchronous cycles.   
\end{lemma}

\begin{proof}
	\noindent \textbf{Invariants (i) and (ii).~~}
	Since $j \in \mathit{Correct}$ and $j \notin r$ in $c$, processor $p_i$ sends the message $\mathsf{MSG}(m,j,s)$ to $p_j$ infinitely often in $R$ (due to the do-forever loop, lines~\ref{ln:forEachMSG} and~\ref{ln:sendMSG} as well as $\mathit{HB}$-liveness and this lemma's assumptions). Moreover, $p_j$ receives $p_i$'s message (line~\ref{ln:sendMSGack}), and acknowledges it, infinitely often, so that $p_i$ receives $p_j$'s acknowledgment (line~\ref{ln:arrMSGack}), infinitely often, while making sure that $j$ is included in $r$ (line~\ref{ln:addJack}). Since a single round-trip is required for the latter to hold, the  cost measure is $2$.
	
	\smallskip
	
	\noindent \textbf{Invariant (iii).~~} This is implied by Remark~\ref{ss:FD asynchronous cycles} applied to the proof of Invariant (ii).
\end{proof}

\subsection{The convergence property} 
Theorem~\ref{thm:recovery} shows that the system reaches a state that satisfies Definition~\ref{def:safeConfig}.

\begin{theorem}[Convergence]
	\label{thm:recovery}
	Let $R$ be a fair execution of Algorithm~\ref{alg:qurb} that starts in an arbitrary system state. Within $\bigO(\buffCapacity)$ asynchronous cycles, the system reaches a state, $c \in R$, after which a suffix $R'$ of $R$ starts, such that $R'$ is complete with respect to the $\mathsf{urbBroadcast}$ invocations in $R$. Moreover, $seq$ and $\msgSet$ are consistent in any $c' \in R'$ (Definition~\ref{def:safeConfig}).
\end{theorem}
\begin{proof}
	The proof is implied by arguments (1) to (6).
	
	\smallskip

	\noindent Argument (1): \emph{The case in which $\mathsf{MSG}(m,\bullet)$ (or its correspondent $\mathsf{MSGack}(\bullet)$) appears in a communication channel at $R$'s starting system state.~~} 
	Suppose that in $R$'s starting system state, it holds that $\mathsf{MSG}(m,\bullet)$ appears in an incoming communication channel to $p_k$. Since $R$ is fair, then by Remark~\ref{ss:first asynchronous cycles} it holds that within $\bigO(1)$ asynchronous cycles, $\mathsf{MSG}(m,\bullet)$, or respectively, $\mathsf{MSGack}(\bullet)$ arrives at its destination, $p_k$. For the case of $\mathsf{MSG}(m,\bullet)$, this arrival results in the execution of line~\ref{ln:MSGarrive} and then line~\ref{ln:addBuff} if $m$'s record was not already in $\msgSet_k$. For the case of $\mathsf{MSGack}(\bullet)$ and $(m,\bullet) \in \msgSet_k$, line~\ref{ln:MSGackarrive} has a similar effect. Moreover, by the code of Algorithm~\ref{alg:qurb}, the case of $\mathsf{MSGack}(\bullet)$ and $(m,\bullet) \notin \msgSet_k$ does not change $p_k$'s state. Therefore, without loss of generality, the rest of the proof can simply focus on the case in which $R$'s starting system state, it holds that $\exists_{p_k \in\sP}:(m,\bullet) \in \msgSet_k$. Note that by similar arguments, we can also consider the case in which the communication channels include message $\mathsf{MSGack}$ with a message identifier $(id,seq)$ that refers to $m$.
	
	\smallskip
	
	\noindent Argument (2): \emph{Definition~\ref{def:safeConfig}'s Invariant (i) holds for the case in which $\exists_{p_k \in\sP}:(m,\bullet) \in \msgSet_k$ in $R$'s starting system state.~~} 
	Within $\bigO(1)$ asynchronous cycles, $p_k$ runs a complete iteration of its do forever loop (lines~\ref{ln:doForever} to~\ref{ln:sendMSG}). Invariant (i) is implied by lines~\ref{ln:noDubOrBot} to~\ref{ln:buffLimit}.

	\smallskip
	
	\noindent Argument (3): \emph{$seq_i$ is greater than or equal to any $p_i$'s sequence values in the variables and fields related to $seq$ in $c'$.~~} 
	Within $\bigO(1)$ asynchronous cycles, every message that was present in a communication channel in $R$'s starting system state arrives at the receiver. Therefore, without loss of generality, the proof can focus on the values of $seq_i$ at the non-failing nodes $p_i,p_k \in \sP$. Other than in $seq_i$, every sequence value that is related to $p_i$ can only be stored in records of the form $(\bullet,id=i,seq=s',\bullet)$ that are stored in $\msgSet_k$. Suppose that in $R$'s starting system state, it holds that $s'> seq_i$. By lines~\ref{ln:URBsendGossip} and~\ref{ln:urbGOSSIPupdate}, within $\bigO(1)$ asynchronous cycles, $p_k$ gossips $s_k\geq s'$ to $p_i$ and the latter updates $seq_i$ upon reception. The argument proof is complete because only $p_i$ (line~\ref{ln:call2urbBroadcast}) can introduce new $seq$ values that are associated with $p_i$, and thus, the argument invariant holds for any system state in $R'$.
	
	\smallskip
	
	\noindent Argument (4): \emph{Definition~\ref{def:safeConfig}'s Invariant (ii) holds in $c'$.~~} 
	The case of $seq$ values is covered by Argument (3). Within an asynchronous cycle, every message arrives at the receiver. Therefore, without loss of generality, the proof can focus on the values of $sMj$ at the non-failing node $p_i \in \sP$. Suppose that in $R$'s starting system state, the predicate $\txS_i[j]\leq \rxS_j[i]$ does not hold for some non-failing $p_j \in \sP:j \in \mathit{trusted}_i$. By lines~\ref{ln:URBsendGossip} and~\ref{ln:urbGOSSIPupdate}, within $\bigO(1)$ asynchronous cycles, $p_j$ gossips $\rxS_j[i]$ to $p_i$ and the latter updates $\txS_i[j]$ upon reception as well as $p_i$ gossips $\txS_i[j]$ to $p_j$ and the latter updates $\rxS_j[i]$ upon reception. Thus, Definition~\ref{def:safeConfig}'s Invariant (ii) holds is any system state that follows.
	
	\smallskip
	
	The rest of the proof assumes, without loss of generality, that Definition~\ref{def:safeConfig}'s invariants (i) and (ii) hold throughout $R$. Generality is not lost due to arguments (1) to (4). 
	
	\smallskip
	
	\noindent Argument (5): \emph{$\forall k \in \mathit{trusted}_i : | \{ (\bullet,id=i,\bullet) \in \msgSet_k \}|\leq \buffCapacity$ holds in $c'$ (first part of Definition~\ref{def:safeConfig}'s Invariant (iii)).~~}
	Let $p_i,p_k \in \sP$ be two non-faulty nodes. 
	For the case of $p_k$'s records in $\msgSet_i$, Definition~\ref{def:safeConfig}'s Invariant (i) says that $\forall (\bullet,id=k,seq=s_k,\bullet) \in \msgSet_i: \max \{ s'_k: (\bullet, id=k,seq=s'_k,\bullet) \in \msgSet_i\} \leq (s_k+\buffCapacity)$. In other words, the largest sequence number, $\max \{ s'_k: (\bullet, id=k,seq=s'_k,\bullet) \in \msgSet_i\}$, of a $p_k$'s records in $\msgSet_i$ minus $\buffCapacity$ must be smaller than $s_k$ of any $p_k$'s records in $\msgSet_i$.
	
	\smallskip
	
	\noindent Argument (6): \emph{$seq_i<\minTxObstSeq_i()+\buffCapacity$ holds in $c'$ (second part of Definition~\ref{def:safeConfig}'s Invariant (iii)).~~}
	%
	%
	Let $c \in R$ and $x_c=(seq_i-\minTxObstSeq())$. Assume, towards a contradiction, that $x_c\geq \buffCapacity$ for at least $\bigO(\buffCapacity)$ asynchronous cycles. Let $A_c=\cup_{k \in \mathit{trusted}_i} \{ r \in \msgSet_k : \neg \obsolete_k(r) \land r.id=i \}$ and $B_c= \cup_{k \in \mathit{trusted}_i} \{ r \in \msgSet_i : r.id=i \land \txS_i[k] < r.seq  \}$ as well as $rec \in A_c$ and $rec' \in B_c$ be the records with the smallest sequence number (among all the records with $id=i$) that $p_k$, and respectively, $p_i$ stores in $c$. We start the proof by showing that, within $\bigO(1)$ asynchronous cycles, the system reaches a state $c' \in R$ for which $rec \notin A_{c'}$ and $rec' \notin B_{c'}$ hold. We then show that $x_{c'}<\buffCapacity$ holds.
	
	\smallskip
	
	\noindent \emph{Showing that $rec \notin A_{c'}$  because $\obsolete(rec)$ holds.~~}
	Let $p_i \in \sP$ be a non-faulty node. Suppose that $\exists p_k\in \sP:(\bullet,id=i,\mathit{delivered}=d_k,recBy=r_k, \bullet) \in \msgSet_k \land k \in \mathit{trusted}_i \land d_k=\false$ holds for some value of $r_k$ in $c$. For any $p_j \in \sP$ for which $j \in \mathit{trusted}_k$ holds throughout $R$'s first $\bigO(1)$ asynchronous cycles, we know that the system reaches, within $\bigO(1)$ asynchronous cycles, a state in which $j \in r_k$ is true  (invariants (i) and (ii) of Lemma~\ref{thm:basic}). Once $\forall j \in \mathit{trusted}_k:j \in r_k$ holds, $p_k$ assigns $\true$ to $d_k$ (line~\ref{ln:trustedSubsetR}). Thus, $rec \notin A_{c'}$ holds within $\bigO(1)$ asynchronous cycles in $R$. Moreover, $\obsolete(rec)$ holds due to the choice of $rec$ as the one with the smallest sequence number.
	
	\smallskip
	
	\noindent \emph{Showing that $rec' \notin B_{c'}$.~~}
	As long as $rec' \in B_{c}$, node $p_i$ sends $\mathsf{MSG}(rec'.msg,rec'.id,rec'.seq)$ to $p_k$ infinitely often (line~\ref{ln:sendMSG}). Within an asynchronous cycle, $p_k$ receives this $\mathsf{MSG}$ message. Lines~\ref{ln:sendMSGack} and~\ref{ln:updateReturn} imply that either $(rec'.msg,rec'.id,rec'.seq,\bullet) \in A_{c'}$ or $rec'.seq \leq \rxS_k[i]$. Within $\bigO(1)$ asynchronous cycles, $(rec'.msg,rec'.id,rec'.seq,\bullet) \notin A_{c'}$ and $\obsolete(rec)$ hold (due the $rec \notin A_{c'}$ case) as well as $\rxS_k[rec'.id]\geq rec'.seq$ (line~\ref{ln:recBytrustedBuffer}). Moreover, by Argument (4)'s proof, $rec' \notin B_{c'}$ since $\txS_k[rec'.id]\geq rec'.seq$ holds.  
	
	\smallskip
	
	\noindent \emph{Showing that $x_{c'}<\buffCapacity$.~~}
	Due to the assumption at the start of this proof, throughout $R$'s first $\bigO(\buffCapacity)$ asynchronous cycles, $p_i$ does not increment $seq_i$ and call $\mathsf{update}()$ (line~\ref{ln:call2urbBroadcast}). Thus, on the one hand, no new $p_i$'s record is added to $\msgSet_k$ throughout $R$'s first $\bigO(\buffCapacity)$ asynchronous cycles (due to the assumption that appears in the start of this case), while on the other hand, within $\bigO(1)$ asynchronous cycles, the system reaches a state in which either $p_i$ stops including $p_k$ in $\mathit{trusted}_i$ or it removes at least one record from $A_{c}$ and $B_{c}$. The latter can repeat itself at most $\buffCapacity$ times due to Argument (5). This completes the proof of the argument and the proof of the theorem.
\end{proof}


\subsection{The closure property} 
Theorem~\ref{thm:closure} considers system executions that reach suffixes, $R$, that satisfy definitions~\ref{def:safeConfig} and~\ref{def:complete}. Theorem~\ref{thm:closure} then shows that $R$ satisfies Definition~\ref{def:URB}, \ie $R \in LE_{\text{URB}}$ is a legal execution.

\begin{theorem}[Closure]
	\label{thm:closure}
	Let $R$ be an execution of Algorithm~\ref{alg:qurb} that is complete with respect to $\mathsf{urbBroadcast}$ invocations (or $R$ is a suffix of an execution $\mathcal{R}=\mathcal{R}'\circ R$ for which $R$ is complete with respect to the invocation of $\mathsf{urbBroadcast}$ messages in $\mathcal{R}$, cf. Definition~\ref{def:complete}) and $seq$'s sequence values are consistent in $c \in R$ (Definition~\ref{def:safeConfig}). Algorithm~\ref{alg:qurb} demonstrates in $R$ a construction of the URB communication abstraction. Moreover, each invocation of the operation $\mathsf{urbBroadcast}()$ incurs $\bigO(n^2)$ messages and has the cost measure of $2$ (Section~\ref{ss:asynchronousCycles}).
\end{theorem}
\begin{proof}
	We note that the property of validity holds with respect to Algorithm~\ref{alg:qurb} due to this theorem's assumption about $R$'s completeness, which can be made due to Theorem~\ref{thm:recovery}. 
	We observe that Algorithm~\ref{alg:qurb} guarantees, within $\bigO(1)$ asynchronous cycles, the property of integrity since only non-delivered messages can be delivered (line~\ref{ln:trustedSubsetR}) and once delivered they cannot be delivered again (line~\ref{ln:deliveredTrue}). 
	
	\smallskip
	
	The proof demonstrates claims~\ref{thm:diffuseAll} and~\ref{thm:urbDeliverAll} before showing the termination and quiescent properties. Let $R$ be an execution of Algorithm~\ref{alg:qurb}. We note that no processor removes the processor identity $\ell$ from the set $recBy$ throughput $R$.
	
	\begin{claim}
		\label{thm:diffuseAll}
		Let $p_i,p_j \in \sP$ be two non-faulty nodes. The fact that $\mathit{diffuse}_i(m)$ (Definition~\ref{def:diffuse}) holds for every system state $c \in R$ implies that eventually, $\mathit{diffuse}_j(m)$ holds in $c' \in R$.
	\end{claim}
	\begin{claimProof}	
		The proof is implied by arguments (2) and (3), which consider Argument (1).
		
		\smallskip
		
		\noindent Argument (1): \emph{Only due to lines~\ref{ln:oprurbBCast} and~\ref{ln:sendMSGack} can $\mathit{diffuse}_i(m)$ hold during $R$.~~} 
		Only line~\ref{ln:addBuff} can add $(m,k,s,\bullet):m\neq \bot$ to $\msgSet_i$. This can only happen due to an earlier invocation of operation $\mathsf{urbBroadcast}(m)$ (line~\ref{ln:oprurbBCast}) or $\mathsf{MSG}$ message arrival (line~\ref{ln:sendMSGack}). 
		
		\smallskip
				
		\noindent Argument (2): \emph{Suppose that $R$ includes a system state $c''$ in which $(m,k,s,r, \bullet) \in \msgSet_i:m\neq \bot \land j \in r$ holds. Then, $\mathit{diffuse}_j(m)$ holds in $R$ (possibly before $c''$).~~} 
		Processor $p_i$ adds $j$ to the $recBy$ field value, $r$, only due to the reception of $\mathsf{MSG}(m)$ (or it's correspondent $\mathsf{MSGack}(\bullet)$) from $p_j$ (lines~\ref{ln:MSGarrive} and~\ref{ln:MSGackarrive}). By the theorem's assumption that $m$ does not appear in $R$'s starting system state, it follows that $p_j$ receives $\mathsf{MSG}(m)$. Immediately after the first reception, it holds that $(m,k,s,\bullet) \in \msgSet_j:m\neq \bot$ in $c' \in R$ (due to the execution of line~\ref{ln:MSGarrive} and then line~\ref{ln:addBuff}). Thus, the predicate $\mathit{diffuse}_j(m)$ holds in $c'$.
		
		
		\smallskip
		
		\noindent Argument (3): \emph{Suppose that $R$ includes no system state in which $(m,k,s,r, \bullet) \in \msgSet_i:m\neq \bot \land j \in r$ holds. Then, the predicate $\mathit{diffuse}_j(m)$ holds eventually.~~} 
		By the assumption that $p_j$ is non-faulty and the $\mathit{HB}$-liveness property (Section~\ref{sec:spec}), eventually, $p_i$ does not suspect $p_j$ and thus, it follows that $p_i$ sends infinitely often the message $\mathsf{MSG}(m)$ to $p_j$ (line~\ref{ln:URBsendGossip}). Due to the fair communication assumption, $p_j$ receives $\mathsf{MSG}(m)$ infinitely often from $p_i$. Immediately after the first time in which $p_j$ receives $\mathsf{MSG}(m)$, it holds that $(m,k,s,\bullet) \in \msgSet_j:m\neq \bot$ in $c' \in R$ (due lines~\ref{ln:MSGarrive} and~\ref{ln:addBuff}). Thus, the predicate $\mathit{diffuse}_j(m)$ holds in $c'$.
		%
	\end{claimProof} 
	
	\begin{claim}
		\label{thm:urbDeliverAll}
		Suppose that for any non-faulty $p_i \in \sP$, there is a system state $c_i \in R$, after which $\mathit{diffuse}_i(m)$ holds. Eventually, in $R$, any non-faulty $p_j \in \sP$ raises $\mathsf{urbDeliver}_j(m)$.
	\end{claim}
	\begin{claimProof}
		Note that this claim's assumption that $\mathit{diffuse}_i(m)$ always holds after $c_i$ implies that eventually $(m,\bullet)$ is always included in the foreach loop of line~\ref{ln:forEachMSG}. Moreover, $p_j$ raises the event $\mathsf{urbDeliver}_j(m)$ only when the if-statement condition in line~\ref{ln:trustedSubsetR} holds. By $\mathit{diffuse}()$'s definition, $\exists (msg=m,\bullet,\mathit{delivered}=\false,\mathit{recBy}=r,\bullet) \in \msgSet_i$ for some value of $r$ in every system state in $R$. The rest of the proof shows that $\mathit{trusted}_i \subseteq r$ holds eventually.   
		
		\smallskip
		
		\noindent Argument (1): \emph{$(m,\bullet, \mathit{recBy}=r_i, \bullet) \in \msgSet_i \land \mathit{Correct} \subseteq r_i$ holds.~~} 
		Lemma~\ref{thm:basic} and the claim's assumption that $\mathit{diffuse}_i(m)$ holds in every system state after $c_i$ imply that $j \in r_i$ holds eventfully. Using arguments that are symmetric to the ones above, also $i \in r_j:(m,\bullet, \mathit{recBy}=r_j, \bullet) \in \msgSet_j$ holds eventfully. Thus, $\mathit{Correct} \subseteq r_i$ holds eventually.
		
		\smallskip
		
		\noindent Argument (2): \emph{$(m,\bullet, \mathit{recBy}=r_i, \bullet) \in \msgSet_i \land \mathit{trusted}_i \subseteq r_i$ holds.~~} 
		By the $\Theta$-liveness property, $\mathit{trusted}_i \subseteq \mathit{Correct}$ eventually. From Argument (1), we have that the system eventually reaches a state after which $\mathit{trusted}_i \subseteq \mathit{recBy}_i$ always holds.  
	\end{claimProof} 
	
	\smallskip
	
	\noindent \textbf{Proof of the termination property.~~}
	Let $p_i,p_j \in \sP$ be two non-faulty nodes. The proof is implied by the following arguments (1) and (2) as well as Argument (6) of Lemma~\ref{thm:recovery}.
	
	\smallskip
	
	\noindent Argument (1): \emph{$p_j$ raises $\mathsf{urbDeliver}_j(m)$ eventfully when $p_i$ invokes $\mathsf{urbBroadcast}_i(m)$.~~} 
	Since $i \in \mathit{Correct}$, the invocation of $\mathsf{urbBroadcast}_i(m)$ makes sure that $(m,\bullet) \in \msgSet_i$ in a way that implies $\mathit{diffuse}_i(m)$ (by lines~\ref{ln:call2urbBroadcast} and~\ref{ln:addBuff} since by Theorem~\ref{thm:recovery}'s Argument (3), line~\ref{ln:addJack} is not executed in steps that invoke $\mathsf{urbBroadcast}_i(m)$, which increment $seq_i$). Since $\mathit{diffuse}_i(m)$ holds, $\mathit{diffuse}_k(m)$ holds eventually for any non-failing $p_k\in \sP$ (Claim~\ref{thm:diffuseAll}). This implies that $p_j$ invokes $\mathsf{urbDeliver}(m)$ (Claim~\ref{thm:urbDeliverAll}).
	
	\smallskip
	
	\noindent Argument (2): \emph{Suppose that a (correct or faulty) processor $p_k$ invokes $\mathsf{urbDeliver}_k(m)$. Any non-failing node invokes $\mathsf{urbDeliver}(m)$.~~} 
	Immediately before $p_k$ invokes $\mathsf{urbDeliver}_k(m)$, it holds that $(m,\bullet, \mathit{recBy}=r_k, \bullet) \in \msgSet_k \land \mathit{trusted}_i \subseteq r_k$ (line~\ref{ln:trustedSubsetR}). Due to the $\Theta$-accuracy property, $\exists j \in \mathit{trusted}_k\cap \mathit{Correct}$. By this theorem's assumption about $R$, the only way in which $j \in \mathit{trusted}_k \subseteq r_k$ can hold, is if, before $p_k$ invokes $\mathsf{urbDeliver}_k(m)$, node $p_j$ had received $\mathsf{MSG}(m,\bullet)$ (or its corresponding acknowledgment, $\mathsf{MSGack}(\bullet)$) that was sent by $p_k$ (due to reasons that are similar to the ones that appear in the proof of Argument (2) of Claim~\ref{thm:diffuseAll}). Upon $p_j$'s first reception of $\mathsf{MSG}(m,\bullet)$ (or $\mathsf{MSGack}(\bullet)$), processor $p_j$ stores $m$, such that $\mathit{diffuse}_j(m)$ holds immediately after (by similar reasons that appear in the proof of Argument (1)). Since $\mathit{diffuse}_j(m)$ holds, then $\mathit{diffuse}_\ell(m)$ holds eventually for any non-failing processor $p_\ell\in \sP$ (Claim~\ref{thm:diffuseAll}) and $p_\ell$ invokes $\mathsf{urbDeliver}_\ell(m)$ (Claim~\ref{thm:urbDeliverAll}).
	
	\smallskip
	
	\noindent \textbf{Proof of the quiescence property.~~}
	In the context of self-stabilization~\cite{DBLP:journals/acta/DolevGS99}, quiescence is demonstrated by showing that any call to $\mathsf{urbBroadcast}(m)$ can result only in a finite number of $\mathsf{MSG}(m,\bullet)$ (or corresponding $\mathsf{MSGack}(\bullet)$) messages. Note that the reception of a $\mathsf{MSGack}()$ does not result in the sending of a message and the sending of a $\mathsf{MSGack}()$ is always due to the reception of a $\mathsf{MSG}(m,\bullet)$ message. Moreover, processors that fail eventually, can only send a finite number of messages. Thus, without loss of generality, the proof focuses on the sending of $\mathsf{MSG}(m,\bullet)$ messages by processors that never fail. The non-faulty $p_i \in \sP$ can send $\mathsf{MSG}(m,\bullet)$ message to $p_k \in \sP$ only when $rec=(msg=m,id=j,seq=s,\mathit{delivered},\mathit{recBy}=r,\mathit{prevHB}=e) \in \msgSet_i : (k \notin r \lor (i=j \land s=\txS_i[k]+1)) \land (e<u[k])$ holds (lines~\ref{ln:forEachMSG} and~\ref{ln:sendMSG}), where $u = \mathit{HB}_i$. In other words, once the system reaches $c \in R$ for which $j \in r$ holds and $rec$ is obsolete, $p_i$ stops sending $\mathsf{MSG}(m,\bullet)$ to $p_j$. If this occurs for any $j \in \mathit{Correct}$, the proof is done. Thus, arguments (1) to (3) assume that the system does not reach $c$. 
	
	\smallskip
	
	\noindent Argument (1): \emph{The case of $j \notin \mathit{Correct}$.~~} 
	Within a finite time, $\mathit{HB}_i[j]$ does not increase (the $\mathit{HB}$-completeness property). Thus, the system reaches within a finite time a suffix for which in any state the predicate $(e[j] < u[j])$ (line~\ref{ln:sendMSG}) does not hold, where $u = \mathit{HB}_i$. Thus, $p_i$ does not send $\mathsf{MSG}(m,\bullet)$ messages to $p_j$ during that suffix.
	
	\smallskip
	
	\noindent Argument (2): \emph{The case of $j \in \mathit{Correct}$ and $\neg (i=j \land s=\txS_i[k]+1)$.~~} 
	Suppose, forwards a contradiction, that $p_i$ never stop sending $\mathsf{MSG}(m,\bullet)$ messages to $p_j$ or that $p_j$ never stop sending (the corresponding) $\mathsf{MSGack}(\bullet)$ messages to $p_i$. By Invariant (ii) of Lemma~\ref{thm:basic}, $p_i$ adds $j$ to $r_i:(msg=m,\bullet,\mathit{recBy}=r_i,\bullet) \in \msgSet_i $. By Algorithm~\ref{alg:qurb}'s code, $p_i$ never removes $j$ from $r_i$, and thus, the predicate $j \in r_i$ remains true forever. This contradicts the assumption made in the start of this case, and thus, the quiescence property holds for Algorithm~\ref{alg:qurb}.
	
	\smallskip
	
	\noindent Argument (3): \emph{The case of $j \in \mathit{Correct}$ and $(i=j \land s=\txS_i[k]+1)$.~~} 
	The proof here is by showing that $p_i$ eventually stores a higher value in $\txS_i[k]$ and thus this argument is true due to Argument (2). Assume, towards a contradiction, that $\txS_i[k]=s-1$ always holds. Since the transmission conditions (line~\ref{ln:sendMSGForEach}) always holds, the sender $p_i$ sends $\mathsf{MSG}(m, j,s,\bullet)$ infinitely often to $p_k$. Due to the communication fairness assumption, $p_k$ eventually receives $\mathsf{MSG}(m, j,s,\bullet)$ and makes sure that $\msgSet_k$ stores $m$'s record (line~\ref{ln:sendMSGack}) or that $s \leq \rxS_k[i]$ (line~\ref{ln:updateReturn}). It turns out that even when only the former case holds, eventually the latter case holds. Specifically, by Invariant (ii) of Lemma~\ref{thm:basic} and the proof of the URB-termination property, $\exists rec=(msg=m,id=j,seq=s,\mathit{delivered}=d,\mathit{recBy}=r,\bullet) \in \msgSet_k:d=\true \land \mathit{trusted}_k \subseteq r_k$ and thus $\obsolete(rec)$ holds eventually. But then, it must be that $\rxS_k[i]\geq s$, which in turn implies that $\txS_i[k]\geq s$ (due to lines~\ref{ln:URBsendGossip} and~\ref{ln:urbGOSSIPupdate} and the communication fairness assumption). This is a contradiction with this case assumption that $\txS_i[k]=s-1$. Thus, the argument and Algorithm~\ref{alg:qurb} is quiescent.
\end{proof}

\section{Extension: FIFO Message Delivery}
\label{sec:fifo}
We discuss an extension for Algorithm~\ref{alg:qurb} for ensuring ``First-in, First-out'' (FIFO) message delivery. This extension is marked by Algorithm~\ref{alg:qurb}'s boxed code lines. Our solution uses a well-known approach, which can be found in~\cite{hadzilacos1994modular,DBLP:books/sp/Raynal18}. For the sake of completeness, we bring the definition of the FIFO-URB abstraction before discussing the implementation details and proof. The abstraction includes the operation $\mathrm{fifoBroadcast}(m)$ and the event $\mathrm{fifoDeliver}(m)$. Definition~\ref{def:fifoURB} requires FIFO-URB-broadcast messages to be delivered by their sending orders (per individual sender).

\begin{definition}[FIFO Uniform Reliable Broadcast~\cite{DBLP:books/sp/Raynal18}]
	\label{def:fifoURB}
	Let $R$ be a system execution. We say that the system demonstrates in $R$ a construction of the FIFO-URB communication abstraction if URB-validity, URB-integrity, and URB-termination requirements are satisfied (Definition~\ref{def:URB}) as well as the following property.   
	\begin{itemize}
		\item \textbf{FIFO message delivery.~~} Suppose that $p_i \in \sP$ takes a step that includes a call to $\mathrm{fifoBroadcast}(m)$ and calling to $\mathrm{fifoBroadcast}(m')$ (possibly in another step). No $p_j \in \sP$ raises the event $\mathrm{fifoDeliver}(m')$ before taking raising the event $\mathrm{fifoDeliver}(m)$  (possibly in another step).
	\end{itemize}
\end{definition}

Our solution (Algorithm~\ref{alg:qurb} including the boxed code lines) considers a well-known approach for ensuring the FIFO message delivery, which can be found in~\cite{hadzilacos1994modular,DBLP:books/sp/Raynal18}. We associate each message arrival with a predicate, cf. the boxed part of the if-statement condition of line~\ref{ln:trustedSubsetC}, which is based on Definition~\ref{def:fifoURB}'s requirement. The predicate uses the array $next[1..n]$ (line~\ref{ln:varNext}), which holds at the $j$-th entry the sequence number of the next-to-be-delivered message from $p_j$ that receiver $p_i$ is allowed to FIFO-deliver. As long the predicate does not hold, the message is buffered by the receiving end. The receiver can then FIFO deliver a pending message as soon as Definition~\ref{def:fifoURB}-based predicate holds. The proposed self-stabilizing solution advances the one in~\cite{hadzilacos1994modular,DBLP:books/sp/Raynal18} with respect to recovery after the occurrence of transient faults, which requires also the use of bounded buffer size.

Lines~\ref{ln:maxS},~\ref{ln:rxSmaxSBuffCapacity},~\ref{ln:URBsendGossip} and~\ref{ln:urbGOSSIPupdate} help to deal with the case in which node $p_i$ is a receiver that  holds at  $next_i[j]$ a sequence number that is higher than the sender's sequence number, $seq_j$. This situation can only occur due to a transient fault and the concern here is that $p_i$ might omit up to $next_i[j]-seq_j$ messages broadcast by $p_j$. Algorithm~\ref{alg:qurb} overcomes such concerns by gossiping to the sender's next-to-be-delivered sequence number, cf. the boxed part of line~\ref{ln:maxS}'s code, which is called in line~\ref{ln:URBsendGossip}. Upon the arrival of such gossip messages from $p_j$ to $p_j$, node $p_j$ can assure that $seq_j$ is not smaller then $next_i[j]$. Moreover, line~\ref{ln:rxSmaxSBuffCapacity} helps $p_i$ to make sure that $next_i[j]$ does not refer to an obsolete message, \ie $next_i[j]\leq \rxS_i[j]$. 

\begin{theorem}
	Algorithm~\ref{alg:qurb} is a self-stabilizing construction of a FIFO uniform reliable broadcast
	communication abstraction in any system in which URB can be built. Moreover, the operation of FIFO uniform reliable broadcast has URB's cost measures. 
\end{theorem}
\begin{proof}
	Claims~\ref{thm:recoveryFIFO} and~\ref{thm:fifoClosure} demonstrate the proof.
	
	\begin{claim}[Convergence]
		\label{thm:recoveryFIFO}
		Let $R$ be a fair execution of Algorithm~\ref{alg:qurb} that starts in an arbitrary system state. Within $\bigO(1)$ asynchronous cycles, the system reaches a state, $c \in R$, after which a suffix $R'$ of $R$ starts, such that $R'$ is complete with respect to the $\mathsf{fifoBroadcast}$ invocations in $R$. Moreover, $seq$ and $\msgSet$ are consistent in any $c' \in R'$ (Definition~\ref{def:safeConfig}).
	\end{claim}
	\begin{claimProof} 
	The proof is along the same lines as the one of Theorem~\ref{thm:recovery} with a minor revision of Argument (3). There is a need to consider not only the largest sequence number stored in $\msgSet$ for a $p_k$'s message record, but also the value of $next[k]$. The boxed part of line~\ref{ln:maxS}'s code implies the latter case.
	\end{claimProof}
	
	\begin{claim}[Closure]
		\label{thm:fifoClosure}
		Let $R$ be an execution of Algorithm~\ref{alg:qurb} that is not necessarily fair but it is complete with respect to $\mathsf{fifoBroadcast}$ invocations (or $R$ is a suffix of an execution $\mathcal{R}=\mathcal{R}'\circ R$ for which $R$ is complete with respect to the invocation of $\mathsf{fifoBroadcast}$ messages in $\mathcal{R}$, cf. Definition~\ref{def:complete}). Moreover,  $seq$ and $\msgSet$ are consistent in any $c \in R$ (Definition~\ref{def:safeConfig}). Algorithm~\ref{alg:qurb} demonstrates in $R$ a construction of the FIFO uniform reliable broadcast. 
	\end{claim}	
	\begin{claimProof}
		The proof is implied using the sequence numbers in $seq$ and $next[]$ as well as the URB communication abstraction and its properties, which Theorem~\ref{thm:closure} shows.
	\end{claimProof}	
\end{proof}

\section{Bounded Self-stabilizing Uniform Reliable Broadcast}
\label{sec:bounded}
In this section, we explain how to transform our unbounded self-stabilizing URB algorithm to a bounded one. We note the existence of several such techniques, \eg Awerbuch \etal~\cite{DBLP:conf/infocom/AwerbuchPV94}, Dolev \etal~\cite[Section 10]{DBLP:journals/corr/abs-1806-03498} and Georgiou \etal~\cite{DBLP:conf/podc/GeorgiouLS19}. The ideas presented in these papers are along the same lines. They present a transformation that takes a self-stabilizing algorithm for message passing systems that uses unbounded operation indices and transforms it into an algorithm that uses bounded indices. The transformation uses a predefined maximum index value, say, $\mathrm{MAXINT} = 2^{64}-1$, and it has two phases. \textsf{(Phase A)} As soon as $p_i$ discovers an index that is at least $\mathrm{MAXINT}$, it disables new invocations of operations. \textsf{(Phase B)} Once all non-failing processors have finished processing their operations, the transformation uses an agreement-based global restart for initializing all system variables. After the end of the global restart, all operations are enabled. For further details, please see~\cite{DBLP:conf/infocom/AwerbuchPV94,DBLP:journals/corr/abs-1806-03498,DBLP:conf/podc/GeorgiouLS19}. 


\section{Conclusions}
\label{sec:disc}
We showed how non-self-stabilizing algorithms~\cite{DBLP:conf/wdag/AguileraCT97,hadzilacos1994modular,DBLP:books/sp/Raynal18} for (quiescent) uniform reliable broadcast can be transformed into one that can recover after the occurrence of arbitrary transient faults. This requires non-trivial considerations that are imperative for self-stabilizing systems, such as the explicit use of bounded buffers. To that end, we developed a flow-control scheme that allows our URB solution to serve as a basis for explicitly bounding the buffer size at the application layer. The need to have this new scheme shows that currently there no conclusive evidence for the existence of a meta-self-stabilizing scheme that can transfer any (or large family of) non-self-stabilizing algorithm from the textbooks into a self-stabilizing one. We simply need to study one problem at a time (and its non-self-stabilizing state-of-the-art) until we have an algorithmic toolkit that is sufficiently generic.

\bibliography{referances}

\begin{thebibliography}{10}

\bibitem{DBLP:conf/wdag/AguileraCT97}
Marcos~Kawazoe Aguilera, Wei Chen, and Sam Toueg.
\newblock Heartbeat: {A} timeout-free failure detector for quiescent reliable
  communication.
\newblock In {\em Distributed Algorithms, 11th International Workshop, {WDAG}
  '97, Saarbr{\"{u}}cken, Germany, September 24-26, 1997, Proceedings}, volume
  1320 of {\em Lecture Notes in Computer Science}, pages 126--140. Springer,
  1997.
\newblock \href {https://doi.org/10.1007/BFb0030680}
  {\path{doi:10.1007/BFb0030680}}.

\bibitem{DBLP:journals/siamcomp/AguileraCT00}
Marcos~Kawazoe Aguilera, Wei Chen, and Sam Toueg.
\newblock On quiescent reliable communication.
\newblock {\em {SIAM} J. Comput.}, 29(6):2040--2073, 2000.
\newblock \href {https://doi.org/10.1137/S0097539798341296}
  {\path{doi:10.1137/S0097539798341296}}.

\bibitem{DBLP:conf/wdag/AguileraTD99}
Marcos~Kawazoe Aguilera, Sam Toueg, and Borislav Deianov.
\newblock Revising the weakest failure detector for uniform reliable broadcast.
\newblock In Prasad Jayanti, editor, {\em Distributed Computing, 13th
  International Symposium, Bratislava, Slovak Republic, September 27-29, 1999,
  Proceedings}, volume 1693 of {\em Lecture Notes in Computer Science}, pages
  19--33. Springer, 1999.
\newblock \href {https://doi.org/10.1007/3-540-48169-9\_2}
  {\path{doi:10.1007/3-540-48169-9\_2}}.

\bibitem{DBLP:series/synthesis/2019Altisen}
Karine Altisen, St{\'{e}}phane Devismes, Swan Dubois, and Franck Petit.
\newblock {\em Introduction to Distributed Self-Stabilizing Algorithms}.
\newblock Synthesis Lectures on Distributed Computing Theory. Morgan {\&}
  Claypool Publishers, 2019.
\newblock \href {https://doi.org/10.2200/S00908ED1V01Y201903DCT015}
  {\path{doi:10.2200/S00908ED1V01Y201903DCT015}}.

\bibitem{DBLP:conf/infocom/AwerbuchPV94}
Baruch Awerbuch, Boaz Patt{-}Shamir, and George Varghese.
\newblock Bounding the unbounded.
\newblock In {\em Proceedings {IEEE} {INFOCOM} '94, The Conference on Computer
  Communications, Thirteenth Annual Joint Conference of the {IEEE} Computer and
  Communications Societies, Networking for Global Communications, Toronto,
  Ontario, Canada, June 12-16, 1994}, pages 776--783. {IEEE} Computer Society,
  1994.
\newblock \href {https://doi.org/10.1109/INFCOM.1994.337661}
  {\path{doi:10.1109/INFCOM.1994.337661}}.

\bibitem{DBLP:journals/spe/Birman99}
Kenneth~P. Birman.
\newblock A review of experiences with reliable multicast.
\newblock {\em Softw., Pract. Exper.}, 29(9):741--774, 1999.

\bibitem{DBLP:conf/netys/BlanchardDBD14}
Peva Blanchard, Shlomi Dolev, Joffroy Beauquier, and Sylvie Dela{\"{e}}t.
\newblock Practically self-stabilizing {Paxos} replicated state-machine.
\newblock In {\em {NETYS}}, volume 8593 of {\em LNCS}, pages 99--121. Springer,
  2014.

\bibitem{DBLP:journals/jacm/ChandraT96}
Tushar~Deepak Chandra and Sam Toueg.
\newblock Unreliable failure detectors for reliable distributed systems.
\newblock {\em J. {ACM}}, 43(2):225--267, 1996.
\newblock \href {https://doi.org/10.1145/226643.226647}
  {\path{doi:10.1145/226643.226647}}.

\bibitem{DBLP:journals/jpdc/DelaetDNT10}
Sylvie Dela{\"{e}}t, St{\'{e}}phane Devismes, Mikhail Nesterenko, and
  S{\'{e}}bastien Tixeuil.
\newblock Snap-stabilization in message-passing systems.
\newblock {\em J. Parallel Distrib. Comput.}, 70(12):1220--1230, 2010.
\newblock \href {https://doi.org/10.1016/j.jpdc.2010.04.002}
  {\path{doi:10.1016/j.jpdc.2010.04.002}}.

\bibitem{DBLP:journals/cacm/Dijkstra74}
Edsger~W. Dijkstra.
\newblock Self-stabilizing systems in spite of distributed control.
\newblock {\em Commun. {ACM}}, 17(11):643--644, 1974.
\newblock \href {https://doi.org/10.1145/361179.361202}
  {\path{doi:10.1145/361179.361202}}.

\bibitem{DBLP:books/mit/Dolev2000}
Shlomi Dolev.
\newblock {\em Self-Stabilization}.
\newblock {MIT} Press, 2000.

\bibitem{DBLP:journals/acta/DolevGS99}
Shlomi Dolev, Mohamed~G. Gouda, and Marco Schneider.
\newblock Memory requirements for silent stabilization.
\newblock {\em Acta Inf.}, 36(6):447--462, 1999.
\newblock \href {https://doi.org/10.1007/s002360050180}
  {\path{doi:10.1007/s002360050180}}.

\bibitem{DBLP:journals/corr/abs-1806-03498}
Shlomi Dolev, Thomas Petig, and Elad~Michael Schiller.
\newblock Self-stabilizing and private distributed shared atomic memory in
  seldomly fair message passing networks.
\newblock {\em CoRR}, abs/1806.03498, 2018.
\newblock URL: \url{http://arxiv.org/abs/1806.03498}, \href
  {http://arxiv.org/abs/1806.03498} {\path{arXiv:1806.03498}}.

\bibitem{DBLP:conf/podc/GeorgiouLS19}
Chryssis Georgiou, Oskar Lundstr{\"{o}}m, and Elad~Michael Schiller.
\newblock Self-stabilizing snapshot objects for asynchronous failure-prone
  networked systems.
\newblock In Peter Robinson and Faith Ellen, editors, {\em Proceedings of the
  2019 {ACM} Symposium on Principles of Distributed Computing, {PODC} 2019,
  Toronto, ON, Canada, July 29 - August 2, 2019.}, pages 209--211. {ACM}, 2019.
\newblock Also appeared in the proceedings of the 7th International Conference
  on Networked Systems {NETYS} as well as a technical report in {CoRR}
  {abs/1906.06420}.
\newblock \href {https://doi.org/10.1145/3293611.3331584}
  {\path{doi:10.1145/3293611.3331584}}.

\bibitem{hadzilacos1994modular}
Vassos Hadzilacos and Sam Toueg.
\newblock A modular approach to fault-tolerant broadcasts and related problems.
\newblock Technical report, Cornell University, Ithaca, NY, USA, 1994.

\bibitem{DBLP:conf/icdcn/ImbsMPR18}
Damien Imbs, Achour Most{\'{e}}faoui, Matthieu Perrin, and Michel Raynal.
\newblock Set-constrained delivery broadcast: Definition, abstraction power,
  and computability limits.
\newblock In {\em 19th Distributed Computing and Networking, {ICDCN}}, pages
  7:1--7:10. {ACM}, 2018.
\newblock \href {https://doi.org/10.1145/3154273.3154296}
  {\path{doi:10.1145/3154273.3154296}}.

\bibitem{DBLP:journals/cn/Lamport78}
Leslie Lamport.
\newblock The implementation of reliable distributed multiprocess systems.
\newblock {\em Computer Networks}, 2:95--114, 1978.
\newblock \href {https://doi.org/10.1016/0376-5075(78)90045-4}
  {\path{doi:10.1016/0376-5075(78)90045-4}}.

\bibitem{DBLP:journals/cacm/Lamport78}
Leslie Lamport.
\newblock Time, clocks, and the ordering of events in a distributed system.
\newblock {\em Commun. {ACM}}, 21(7):558--565, 1978.
\newblock \href {https://doi.org/10.1145/359545.359563}
  {\path{doi:10.1145/359545.359563}}.

\bibitem{DBLP:conf/hase/Raynal97}
Michel Raynal.
\newblock A case study of agreement problems in distributed systems:
  Non-blocking atomic commitment.
\newblock In {\em 2nd High-Assurance Systems Engineering Workshop {(HASE} '97),
  August 11-12, 1997, Washington, DC, USA, Proceedings}, pages 209--214. {IEEE}
  Computer Society, 1997.
\newblock \href {https://doi.org/10.1109/HASE.1997.648067}
  {\path{doi:10.1109/HASE.1997.648067}}.

\bibitem{DBLP:books/sp/Raynal18}
Michel Raynal.
\newblock {\em Fault-Tolerant Message-Passing Distributed Systems - An
  Algorithmic Approach}.
\newblock Springer, 2018.
\newblock \href {https://doi.org/10.1007/978-3-319-94141-7}
  {\path{doi:10.1007/978-3-319-94141-7}}.

\bibitem{DBLP:journals/csur/Schneider90}
Fred~B. Schneider.
\newblock Implementing fault-tolerant services using the state machine
  approach: {A} tutorial.
\newblock {\em {ACM} Comput. Surv.}, 22(4):299--319, 1990.
\newblock \href {https://doi.org/10.1145/98163.98167}
  {\path{doi:10.1145/98163.98167}}.

\end{thebibliography}
\end{document}